\documentclass[10pt,journal,compsoc]{IEEEtran}

\ifCLASSINFOpdf
\else
\fi
\usepackage{ifpdf}
\ifodd 0
\newcommand{\wei}[1]{{\color{blue}#1}}
\else
\newcommand{\wei}[1]{#1}
\fi
\ifodd 0
\newcommand{\rev}[1]{{\color{blue}#1}}
\else
\newcommand{\rev}[1]{#1}
\fi
\ifodd 0
\newcommand{\rtwo}[1]{{\color{blue}#1}}
\else
\newcommand{\rtwo}[1]{#1}
\fi

\usepackage{color}
\usepackage[cmex10]{amsmath}
\usepackage{bm}
\usepackage{amstext}
\usepackage{amsmath}
\usepackage{amssymb}
\usepackage{amsthm}
\usepackage{multirow}
\usepackage{mathrsfs}
\usepackage{graphicx}
\usepackage{algorithm} 
\usepackage{algorithmic}
\usepackage{harpoon} 
\newtheorem{proposition}{Proposition}
\newtheorem{lemma}{Lemma}
\newtheorem{problem}{Problem}
\newtheorem{theorem}{Theorem}
\newtheorem{definition}{Definition}
\newtheorem{corollary}{Corollary}

\def\ud{\mathrm{d}}

\ifCLASSOPTIONcompsoc
  \usepackage[nocompress]{cite}
\else
  \usepackage{cite}
\fi

\ifCLASSINFOpdf
\else
\fi

\begin{document}

\title{A Novel Mobile Data Contract Design with Time Flexibility}

\author{Yi~Wei,~\IEEEmembership{Student Member,~IEEE,}
        Junlin~Yu,~\IEEEmembership{Student Member, ~IEEE,}
        Tat M.~Lok,~\IEEEmembership{Senior~Member, ~IEEE,}
        and~Lin~Gao,~\IEEEmembership{Senior~Member,~IEEE}
\IEEEcompsocitemizethanks{
\IEEEcompsocthanksitem Yi Wei, Junlin Yu and Tat M. Lok are with the Department
of Information Engineering, The Chinese University of Hong Kong, Shatin, N.T.,
Hong Kong. Email: {\{wy012, yj112, tmlok\}@ie.cuhk.edu.hk}
\IEEEcompsocthanksitem
Lin Gao (corresponding author) is with the Department of Electronic and Information Engineering, Harbin Institute of Technology, Shenzhen, China.
Email: {gaol@hit.edu.cn}

Part of the results have appeared in IEEE ICCS 2016 \cite{7833585}.
}}


\IEEEtitleabstractindextext{%
\begin{abstract}
In conventional mobile data plans, the data is associated with a fixed period (e.g., one month) and the unused data will be cleared at the end of each period.
To take advantage of consumers' heterogeneous demands across different periods and meanwhile to provide more time flexibility, some mobile data service providers (SP) have offered data plans with different lengths of period. In this paper, we consider the data plan design problem for a single SP, who provides data plans with different lengths of period for consumers with different characteristics of data demands.
We propose a contract-theoretic approach, wherein the SP offers a period-price data plan
contract which consists of a set of period and price combinations, indicating the prices for data with different periods.
We study the optimal data plan contract designs under two different models: discrete and continuous consumer-type models, depending on whether the consumer type is discrete or continuous.
In the former model, each type of consumers are assigned with a specific period-price combination.
In the latter model, the consumers are first categorized into a finite number of groups, and each group of consumers (possibly with different types) are assigned with a specific period-price combination.
We systematically analyze the incentive compatibility (IC) constraint and individual rationality (IR) constraint, which ensure each consumer to choose the data plan with the period-price combination intended for his type.
We further derive the optimal contract that maximizes the SP's expected profit, meanwhile satisfying the IC and IR constraints of consumers.
Our numerical results show that the proposed optimal contract can increase the SP's profit by $35\%$, comparing with the conventional fixed monthly-period data plan.
\end{abstract}

}

\maketitle

\IEEEdisplaynontitleabstractindextext

\IEEEpeerreviewmaketitle

\ifCLASSOPTIONcompsoc
\IEEEraisesectionheading{\section{Introduction}\label{sec:introduction}}
\else
\section{Introduction}
\label{sec:introduction}
\fi

\subsection{Background and Motivation}
The fast development and wide adoption of smart phones and tablet devices not only drive the explosive growth of mobile data consumption, but also increase the consumption fluctuation over different plan periods \cite{joe2015mobile}. Conventionally, each data plan specifies data cap for a specific period, which is usually one month. The unused data will be cleared at the end of the period, and the overused data will be charged an additional fee. Hence, the consumers with large consumption fluctuation over different periods (e.g., those having frequent trips) will suffer a large utility loss, because the overused data cannot be compensated by the leftover data in the previous periods.

To deal with this problem, researchers in both academia and industry have proposed many data pricing schemes \cite{sen2013survey, 6849296, 6848090, 6562872, 6849295, 7218537, 7151094}.
However, these pricing schemes do not
fully take advantage of users'  heterogeneous demands across periods.
Seizing this opportunity, some major mobile data service providers (SP) including AT\&T \cite{TP-toolbox-web} and T-mobile \cite{Tmobile} have launched a novel data plan called \emph{rollover data plan}, where unused data from the monthly plan allowance rolls over for \emph{one} billing period.

The rollover data plan provides customers more time flexibility by decreasing the frequency of clearing unused data from once per month to once every two months.
However, the two-month's time flexibility may not be enough for the consumers with highly varying data demand.
\rev{In order to provide time flexibility to more types of customers, we propose a new type of data plan, which specifies the length of period.
As a simple example, an SP can offer multiple data plans to consumers, e.g., $1$GB for every month with a price of \$$30$ per month, $6$GB for every six months (i.e., with average data cap $1$GB per month) with a price of \$$32$ per month, and $12$GB for every year (i.e., with average data cap $1$GB per month) with a price of \$$35$ per month.}
Such data plans can benefit different types of consumers.
On one hand, the consumers with highly varying data demand may prefer the data plan with a long period (which provides more time flexibility and can potentially reduce the uncertainty of data demand).
On the other hand, the consumers with rarely varying data demand may prefer the data plan with a smaller period (which can reduce the total cost due to the lower unit price).
In such a scenario, a natural problem for the SP is how to design a proper set of data plans to maximize its expected profit.
The problem is challenging due to (a) the information asymmetry between the SP and consumers  and (b) the difficulty in discriminating consumers.

\subsection{Key Results and Contributions}

In the first part of this paper, we propose a contract-theoretic mechanism for a single SP for discrete-consumer-type model.
The SP offers a \emph{contract} consisting of a set of period-price combinations, where each period-price combination is designed for a specific type of consumers with a specific data demand distribution.
Contract theory has been widely applied in solving economics, marketing and network problems \cite{6342942} \cite{6464648}, and is a useful tool in designing incentive compatible (IC) and individual rational (IR) mechanism \cite{bolton2005contract} to elicit the private information of end users.
In this work, we adopt the contract theory to solve the SP's profit maximization problem under information asymmetry.
Specifically, we first provide the IC and IR constraints for the feasible contract to guarantee the truthful demand information revelation of consumers, based on which we further derive the optimal contract that maximizes the SP's expected profit.

In the second part of the paper, we extend our study to a more general mechanism, which is designed for continuous-consumer-type model.
In this case, providing a period-price combination for each consumer type is equivalent to providing infinite combinations, which is not realistic and not consumer friendly.
Therefore, our mechanism for continuous consumer types includes the procedure of dividing users into groups according to their types.
Then, we design limited pairs of period-price combinations, where each combination is designed for a group of consumers.
It is very challenging to use a limited period-price combination to model the infinite consumer types, which usually leads to an NP-hard problem \cite{li2014dynamic}.
Therefore, an alternative maximizing algorithm is introduced to find a sub-optimal solution.
In the algorithm, we alternatively update the period assignments and group boundaries in order to maximize the SP's total profit.
The main challenge of this method lies in the step of updating group boundaries with fixed period assignment due to the non-convexity of the problem.
However, by exploiting the unimodal structure of the objective function, we can obtain the sufficient condition for the optimal solution and show that sufficient condition is satisfied for different scenarios.

The main contributions of the paper are as follow.
\begin{enumerate}
    \item \emph{Novel Model:} We study the SP's mobile data plan design problem from the perspective of data period, which provides consumers with more time flexibility.
    To our best knowledge, this is the first work that systematically studies such a new data plan design perspective.
	\item \emph{Novel Method:} We propose a novel period-price data plan based on the contract theory.
	Rather than specifying a price for each data cap in conventional data plans, our proposed data plan specifies a data price for each data period. A higher price is associated with a longer data period.
	\item \emph{Systematic Solution:} \rev{We first analyze the period-price contract for discrete-consumer-type model, and then extend the analysis to a more general model with continuous-consumer-type.
	In continuous-consumer-type model, we assume that the consumer type follows a continuous distribution but the SP offers only a limited number of contract items, which is different from traditional continuous modeling in contract theory.
	In both cases, we analyze the feasibility (incentive compatibility and individual rationality) of the proposed period-price contract systematically, based on which we further derive the optimal contract that maximizes the SP's profit.}
	\item \emph{Performance Evaluation:} We compare our proposed optimal contract with the conventional monthly-period scheme through numerical simulations. Numerical results show that our proposed contract can increase the SP's profit over $35\%$.
\end{enumerate}

\subsection{Related Literature on Data Pricing Schemes}

The survey by Sen \emph{et al.} in \cite{sen2013survey} reviewed the past pricing proposals and discussed several potential research problems. There are mainly three categories of methods to alleviate the problem of monthly data plan inflexibility: (1) \emph{Shared Data Plan} \cite{6849296} \cite{6848090} allows sharing data quota among multiple devices or users, and hence to decrease the average unit usage cost.
(2) \emph{Sponsored Data} \cite{6562872} \cite{6849295} is offered by the content service providers, to sponsor the end users for the traffic of viewing their content.
(3) \emph{Secondary Data Trading} \cite{7218537} \cite{7151094} is proposed by the service providers, which allows users to trade their unused mobile data with each other.
However, all the above methods have their own disadvantages. Shared data plan does not fully take advantage of the heterogeneous demands across plan periods, because there exists possibility that everyone in the shared data is in the peak month. Sponsored data is too specific to the contents, because not every content provider is willing to provide this sponsorship. Secondary data trading is not convenient for operation, since the consumer has to buy or sell every time when he is running out of data or has data left unused.

The papers \cite{7562159,shidi,shidi2,add-1,add-2} are the pioneer works that study the rollover data plan. Zheng \emph{et al.} in \cite{7562159} evaluated the benefits of rollover data for both SPs and users as well as identify the types of users who would upgrade to rollover data plans.
Wang \emph{et al.} in \cite{shidi} and \cite{shidi2} analyzed the interactions between an SP and its subscribed users under both traditional and rollover
data plans.
In \cite{add-1} and \cite{add-2}, they further analyzed the competitive market with multiple SPs offering rollover data plans with fixed rollover period (i.e., one month).
However, none of them considers the design of data plans from the dimension of length of period.
To the best of our knowledge, this work is the first paper that systematically studies a data plan design regarding the length of period.

The remainder of this paper is organized as follows.
We first analyze the optimal contract for discrete-consumer-type model in Section \ref{sec:systemmodel} and Section \ref{sec:contractfando}.
Specifically, we present the system model and formulate the problem in Section \ref{sec:systemmodel}.
We analyze the feasibility of the contract and propose the optimal contract in Section \ref{sec:contractfando}.
We analyze the generalized contract for continuous-consumer-type model with group division in Section \ref{sec:groupdivision}.
Performance evaluation is illustrated in Section \ref{sec:simulation}.
Finally, Section \ref{sec:conclusion} concludes the paper.

\section{System Model}
\label{sec:systemmodel}

\subsection{Service Provider Modeling}
In the conventional data plans, the SP provides a unique period choice (e.g., one month). In those data plans, each consumer can consume data up to a quantity of $q$ during one period. In our proposed data plans, the SP offers multiple data plans with different plan periods, and we denote the length of the period as $t$ ($t \in (0,+\infty)$). \footnote{
For presentation convenience, in the rest of the paper, we use ``period $t$" to refer to ``period with length $t$", and use ''unit period" to refer to ''period with length 1".
We assume that $t$ can be any positive number, so that it is possible to provide any time flexibility.
}

To sharpen the insights of plan periods, we assume that all the data plans are with the same data cap $q$ for a unit time period.
Figure 1 is an example of the contract of data plans provided by the SP.
In the contract, the SP offers a set of combinations, where each combination (also called a contract item) corresponds to one data plan.
In each combination, there is a period $t$ and a corresponding unit period \emph{price} $\pi(t)$.
In other words, a consumer who chooses the contract item $\{t,\pi(t)\}$, needs to pay a price $t\pi(t)$, and can consume data up to a quantity of $tq$ in a period of $t$.
Intuitively, the unit period price $\pi(t)$ is an increasing function of $t$, because a larger period provides more time flexibility for consumers.

\begin{figure}[tbp]
	\centering
		\includegraphics[width=85mm]{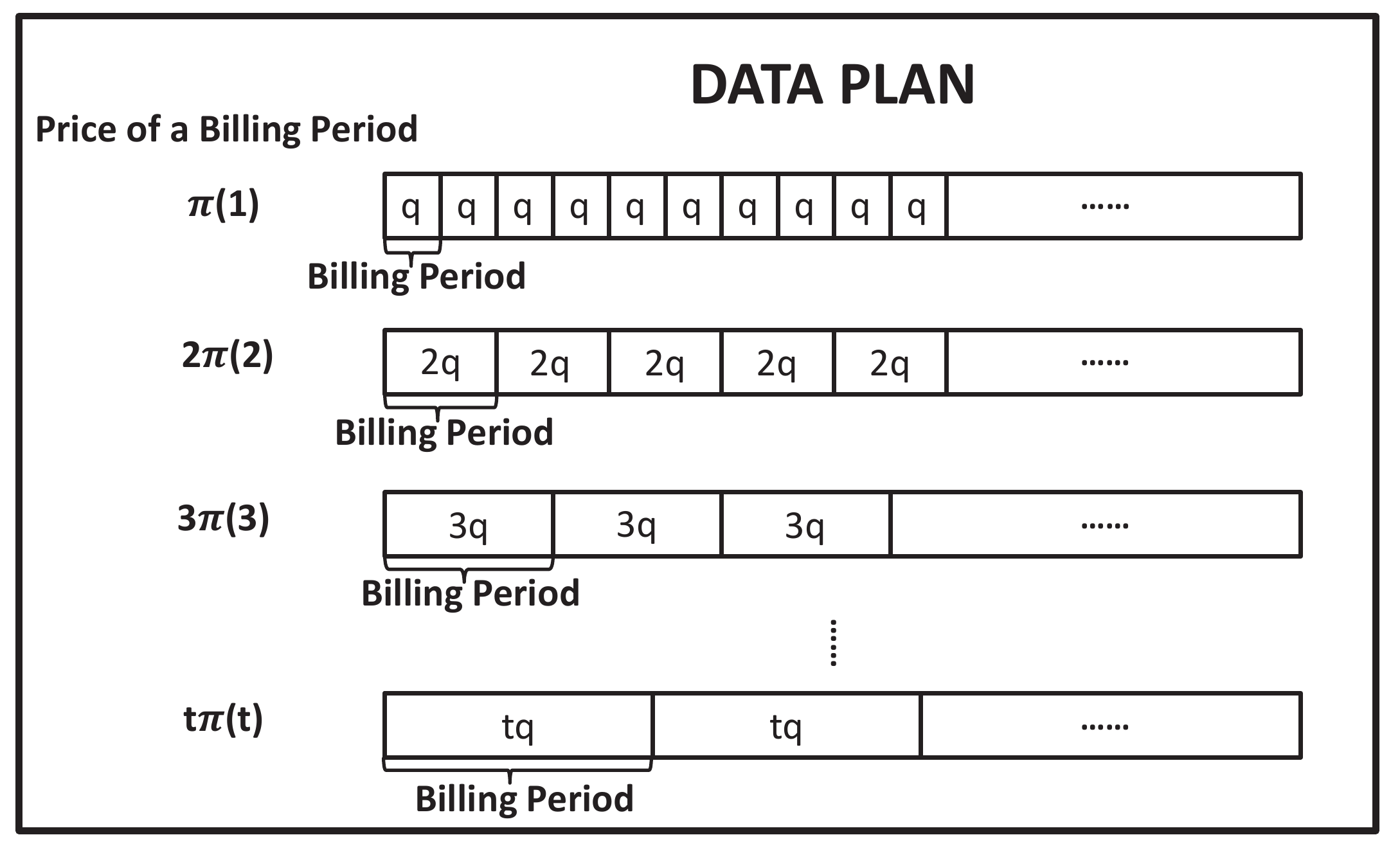}
	\caption{An example of the contract of data plans.}
	\label{fig:dataplan}
	\vspace{-2mm}
\end{figure}

We define the \emph{cost} for the SP as the average expense of providing a data plan of period $t$ with data cap $tq$.
Here, we denote the total expense of a data plan of period $t$ as $tC(t)$, where the cost (i.e., average expense for one period) $C(t)$ is formulated as follows
\begin{equation}
C(t) = W(t) + C_0,\notag
\end{equation}
where $C_0$ is the fixed cost (e.g., the fixed monthly spectrum license fee for providing service and the infrastructure maintenance cost) and $W(t)$ is the time-specific cost. We can show that $W(t)$ is monotone increasing on $t$. This is because with a smaller $t$, the SP can better predict and hence schedule the demand of consumers. We further assume that $W(t)$ grows more rapidly with larger period $t$, which means $W_{tt}(t)\geq 0$.\footnote{
$g_{x}(.)$ denotes the first order derivative of $g(.)$ \rtwo{with} respect to $x$ ($\partial g(.)/\partial x$). $g_{xx}(.)$ and $g_{xy}(.)$ denotes $\partial^2 g(.)/\partial x^2$ and $\partial^2 g(.)/\partial x \partial y$ respectively.
}
Then, we can see $C_t(t) > 0$ and $C_{tt}(t) \geq 0$.

The SP's \emph{profit} comes from selling data plans with different periods. We use $R(t)$ to denote the unit period profit of SP from the data plan $\{t, \pi(t)\}$, which is the gap between the unit period price $\pi(t)$ and the unit period cost $C(t)$, i.e.,
\begin{equation}
R(t) = \pi(t) - C(t).\notag
\end{equation}

\subsection{Consumer Modeling}
\label{sec:consumermodeling}
We assume consumer $i$'s data demand per unit period follows normal distribution with density function $f(x|\mu_i, \sigma_i)$, where $\mu_i$ is the mean and $\sigma_i$ is the standard deviation \footnote{
\rev{Normal distribution is widely adopted in modeling consumers' demand in different areas.
For example, \cite{b13} applied normal distribution to model consumers' connectivity of mobile network, while \cite{b14} applied normal distribution to model consumers' demand of electricity.}}.
\rev{In our paper, we focus on the demand fluctuation over unit periods, so our design of contract is based on $\sigma$, i.e., standard deviation of consumers' data demand per unit period.
Therefore, in our modeling, we assume that consumers are divided into different groups with different average monthly demand $\mu$, and we only design contract for a particular group with a certain $\mu$.
Hence, we assume that every consumer has the same $\mu$ and different $\sigma$.}
For writing convenience, we call a consumer as a type-$\sigma$ consumer if the standard deviation of his data demand is $\sigma$.
We first assume that the consumer types follow a discrete distribution, and the SP aims to design a specific contract item for each consumer type.\footnote{
Mathematically, when we choose a large enough number of types, the discrete-consumer-type model can well approximate a continuous-consumer-type model.
In reality, since consumers are heterogeneous, it is more reasonable to assume that consumers' types follow a continuous distribution \cite{a2}. Hence, we will introduce the case of continuous consumer type distribution in Section \ref{sec:groupdivision}.
}
We denote the set containing all consumer types as $\Sigma$.
\rev{Due to the properties of the normal distribution, when a period of the data plan is changed from $1$ unit period to a period of $t$, a consumer's total data demand within a period still follows normal distribution.
The parameters of this normal distribution are as follows: the mean value of the consumer's total data demand within $t$ months is $t\mu$, and the standard deviation is $\sqrt{t}\sigma$.
\footnote{
Intuitively, a consumer with a larger $\sigma$ is with a higher data fluctuation, and naturally needs a data plan with higher time flexibility.
}}

We use $V(\sigma, t)$ to denote the \emph{valuation} of a type-$\sigma$ consumer for the contract with period $t$.
\rev{Similar to \cite{b1}, for a given period $t$ with data cap $tq$, we define the unit period valuation $V(\sigma,t)$ as a linear function on the average data consumption per unit period:}
\begin{align}
V(\sigma,t) \!&=\! \alpha\left(\mu-\frac{1}{t}\int_{tq}^{+\infty}(x-tq)f(x|t\mu,\sqrt{t}\sigma)\ud{x}\right)\notag\\
&=\!\alpha\left(\mu - \frac{1}{t}\int_{\frac{\sqrt{t}\Delta{q}}{\sigma}}^{+\infty}(\sqrt{t}\sigma x-t\Delta{q})f(x|0,1)\ud{x}\right),\label{eq:valuation}
\end{align}
where $\alpha >0$ is a predefined parameter, which represents the valuation of unit data, and is identical for all consumers.
For writing convenience, we define $\Delta{q} = q - \mu$.
We assume the cost of usage exceeding data cap is very large, so that the consumption will not exceed the data cap $tq$.
\rev{Hence, the average data consumption per unit period equals to the consumer's average data demand $\mu$, which is the consumer's maximum average data consumption per unit period, minus average unsatisfied data demand.
In a $t$ period data plan, the total unsatisfied demand of a type-$\sigma$ consumer in a period of $t$ is $\int_{tq}^{+\infty}(x-tq)f(x|t\mu,\sqrt{t}\sigma)\ud{x}$, then the average unsatisfied data demand per unit period is $\frac{1}{t}\int_{tq}^{+\infty}(x-tq)f(x|t\mu,\sqrt{t}\sigma)\ud{x}$.
For example, if a consumer with average data demand $\mu =9$GB and standard deviation $\sigma = 2$ consume monthly data plan with a quota of $10$GB, and his demands of consecutive two months are $11$GB and $7$GB, then his unsatisfied data demand of these two months are $(11-10)^+=1$GB and $(7-10)^+=0$GB, respectively.
Since we assume the consumer's demand per unit period follows normal distribution with $\mu=9$ and $\sigma=2$, the average unsatisfied data demand equals to $\int_{-\infty}^{+\infty}(x-10)^{+}f(x|9,2)\ud{x}=\int_{10}^{+\infty}(x-10)f(x|9,2)\ud{x} = 0.5$GB,
which means his average total data consumption is $9\text{GB}-0.5\text{GB}=8.5\text{GB}$.}
As shown in Figure 2, if all the consumers choose the plan of period $1$ and $q = \mu$, only the consumer with $\sigma = 0$ (the black dashed line) can reach an average consumption of $\mu$.
On the contrary, the consumer with $\sigma = 4$ (the red dashed line) can satisfy his demand in the $1^{st}$, $3^{rd}$, $6^{th}$ and $12^{th}$ periods, but only consumes $q$ in the other periods due to the data cap.

\begin{figure}[tbp]
	\centering
	\vspace{-4mm}
		\includegraphics[width=83mm]{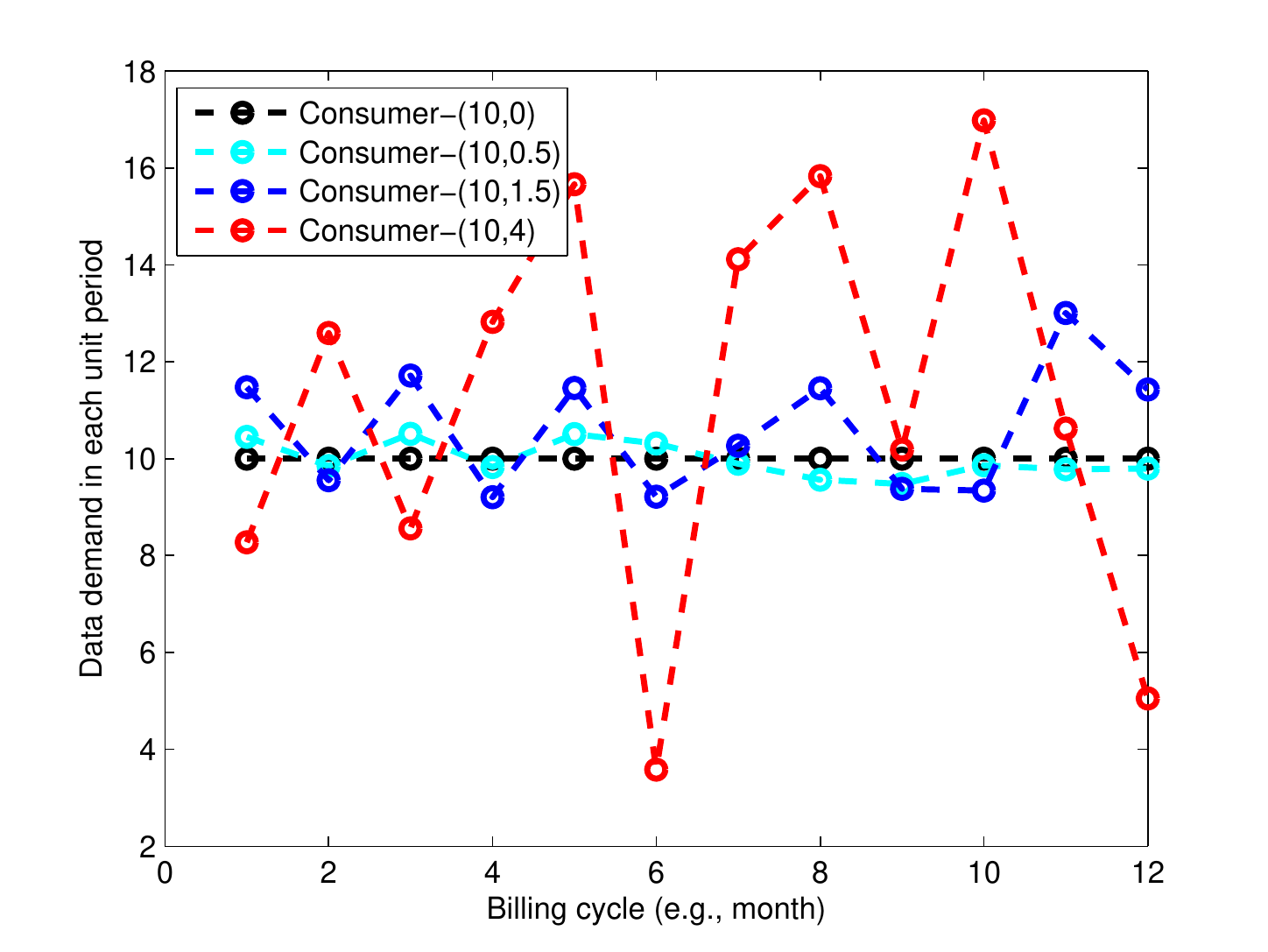}	
		\vspace{-3mm}
	\caption{A data demands example of consumers with different types-($\mu$, $\sigma$).}
	\vspace{-3mm}
	\label{fig:demand}
\end{figure}

From \eqref{eq:valuation}, we can find that
\begin{equation}
V_t(\sigma,t)= \frac{\alpha\sigma}{2t^{1.5}}\int_{\frac{\sqrt{t}\Delta{q}}{\sigma}}^{+\infty}xf(x|0,1)\ud{x}>0\notag
\end{equation}
and
\begin{equation}
V_{\sigma}(\sigma,t)=-\frac{\alpha}{\sqrt{t}}\int_{\frac{\sqrt{t}\Delta{q}}{\sigma}}^{+\infty}xf(x|0,1)\ud{x}<0,\notag
\end{equation}
which means that 1) without considering price, every consumer prefers a larger period and 2) the consumer with larger type has a smaller valuation.
Furthermore, $V_{tt}(\sigma,t)$
is negative through direct calculation, meaning that $V(\sigma,t)$ grows more slowly in a larger period.

The \emph{utility} of the consumer with type-$\sigma$ who accepts the data plan with period $t$ is defined as the gap between his valuation and payment of the data plan:
\begin{equation}
U(\sigma,t) = V(\sigma,t) - \pi(t).\notag
\end{equation}

\subsection{Contract Formulation}
In this paper, we aim to design an optimal contract for the SP to maximize its expected profit. The contract contains a set of combinations, each of which includes a period $t$ and a corresponding unit period price $\pi(t)$.
Each consumer can only select one combination. Therefore, for each consumer type $\sigma \in \Sigma$, the SP will assign a period $t(\sigma)$ with unit period price $\pi(t(\sigma))$. The set of period-price combinations shown above is a \emph{period-price contract}. We denote the contract as $\mathscr{C}_d = \{\big(t(\sigma),\pi(t(\sigma))\big)|~\forall \sigma \in \Sigma\}$.

A feasible contract should satisfy the following two constraints: 1) For any type-$\sigma$ consumer, he prefers the contract item with period $t(\sigma)$ at the price $\pi(t(\sigma))$ than any other contract items; 2) The SP should guarantee that the contract designed for any type-$\sigma$ consumer leads to non-negative utility so that the consumer is willing to accept the contract designed for him. These two constraints are named as incentive compatibility (IC) constraint and individual rationality (IR) constraint correspondingly. Specifically, we define,
\begin{definition}\label{def:definition1}
IC constraint:
\begin{equation}
V(\sigma,t(\sigma))-\pi(t(\sigma)) \geq V(\sigma,t(\sigma'))-\pi(t(\sigma')), ~~\forall \sigma' \neq \sigma.\notag
\end{equation}
\end{definition}
\begin{definition}\label{def:definition2}
IR constraint:
\begin{equation}
V(\sigma,t(\sigma)) - \pi(t(\sigma)) \geq 0,~~\forall \sigma \in \Sigma.\notag
\end{equation}
\end{definition}
\noindent Any feasible contract satisfies IC and IR constraints, and any contract satisfying IC and IR constraints is feasible. The overall profit of the SP from a feasible contract $\mathscr{C}_d = \{t(\sigma), \pi(t(\sigma))|~\forall \sigma \in \Sigma\}$ can be written as:
\begin{equation}
R = \sum_{\sigma \in \Sigma}N_{\sigma}\Big(\pi(t(\sigma))-C\big(t(\sigma)\big)\Big), \label{eq:overallre}
\end{equation}
where $N_{\sigma}$ is the number of consumers with type-$\sigma$.

\begin{table}[tbp]
\centering
\caption{Notation}
\label{table:1}
\begin{tabular}{|c|c|}
\hline
Symbol&Meanings\\
\hline
\hline
$q$ & the average data cap per unit period\\
\hline
$t_i$& the time length of the $i^{th}$ data plan's period\\
\hline
$\pi_i$&the unit period price of the $i^{th}$ data plan\\
\hline
\multirow{2}*{$\{t_i,\pi_i\}$}& the data item with price $t_i\pi_i$\\
 & and data quota $t_iq$ in a period of $t_i$\\
\hline
\multirow{3}*{$C(t_i)$} &the SP's unit period expense of offering data plan\\
&with period $t_i$, i.e., the total expense of a data plan\\
&of period $t_i$ is $t_iC(t_i)$\\
\hline
\multirow{2}*{$R(t_i)$}&the unit period profit of the SP\\
 &from the data plan item $\{t_i, \pi_i\}$\\
\hline
$\mu$&consumers' average data demand in a unit period\\
\hline
\multirow{7}*{$\sigma_i$}& In discrete-consumer-type model:\\
 &$\sigma_i$ represents the standard deviation\\
 & of the unit period data demand of consumer $i$;\\
& In continuous-consumer-type model:\\
& $\sigma_i$ represents the largest standard deviation\\
&among the unit period data demand\\
&of the consumers in the $i^{th}$ group\\
\hline
\multirow{2}*{$V(\sigma, t)$}&the valuation of a type-$\sigma$ consumer\\
& for the contract with period $t$\\
\hline
$N$& the total number of consumers\\
\hline
$N_i$& the number of consumers in group $i$\\
\hline
\multirow{2}*{$g(\cdot)$}&the probability density function\\
& of the distribution of consumer type\\
\hline
\multirow{2}*{$G(\cdot)$}&the cumulative distribution function\\
& of the distribution of consumer type\\
\hline
\end{tabular}
\end{table}

\section{Contract Feasibility and Optimality}
\label{sec:contractfando}
In this section, we first show the necessary and sufficient conditions for the contract to be feasible. Then we derive the best period assignments and price assignments for the optimal contract that maximizes the SP's overall profit, which is defined in \eqref{eq:overallre}.


According to our assumption in Sec. \ref{sec:consumermodeling}, there is a finite number of consumer types $I$. Without loss of generality, we let $\sigma_1 < \sigma_2 < \cdots < \sigma_I$. Then, we rewrite the period $t(\sigma_i)$ assigned to the type-$\sigma_i$ consumers as $t_i$, and rewrite the price $\pi(t(\sigma_i))$ corresponding to the period $t(\sigma_i)$ as $\pi_i$ for simplicity.
Accordingly, we can rewrite the SP's profit function $R(t(\sigma_i))$ and cost function as $C(t(\sigma_i))$ as $R(t_i)$ and cost function as $C(t_i)$, respectively. For   convenience, we summarize the key notations in Table \ref{table:1}.

Therefore, the contract optimization problem can be written as:
\begin{problem}\label{pro:problem1}
\begin{equation}
\begin{aligned}
&\max_{\substack{\{t_i\} \\ \{\pi_i\}}}\ \sum_{i = 1}^{I}N_i(\pi_i - C(t_i)),\\
&~\text{s.t.}\ \left\{
\begin{aligned}
&V(\sigma_i,t_i) - \pi_i \geq V(\sigma_i,t_j) - \pi_j, ~\forall i,j\!\in\! \mathcal{I},  j\!\neq\! i  &\textbf{(IC)}\\
&V(\sigma_i,t_i) - \pi_i \geq 0,~~\forall i \in\mathcal{I} &\textbf{(IR)}\notag\\
\end{aligned}
\right.
\end{aligned}
\end{equation}
\end{problem}
where $\mathcal{I}\!=\!\{1,2,\ldots, I\}$.

\subsection{Feasibility}
According to \eqref{eq:valuation}, we have the following property: for a given period length increment, the consumers with larger type will have a larger valuation increment than the consumers with smaller type. We call this property as increasing preference (IP) property.\footnote{Due to space limit, we put all of the detailed proofs in the online technical report \cite{report}.}

\begin{proposition}[IP property]\label{pop:proposition1}
For any consumer types $\sigma > \sigma'$ and any data plan periods $t > t'$, the following condition holds:
\begin{equation}
V(\sigma,t)-V(\sigma,t') > V(\sigma',t)-V(\sigma',t').\label{eq:ip1}
\end{equation}
\end{proposition}

Now, we try to find the necessary and sufficient conditions for the contract to be feasible, i.e., the necessary and sufficient conditions of IC and IR constraints.
We show the first necessary condition in the following lemma.
\begin{lemma}\label{lm:lemma1}
For any contract $\mathscr{C}_d =\{(t_i,\pi_i)\}$, if it is feasible, then the following condition holds:
\begin{equation}
\sigma_i>\sigma_j \Rightarrow t_i \geq t_j.\notag
\end{equation}
\end{lemma}


Lemma \ref{lm:lemma1} shows that the consumer with larger type $\sigma$ should be assigned a longer period.

We show the second necessary condition in the following lemma.

\begin{lemma}\label{lm:lemma2}
For any contract $\mathscr{C}_d =\{(t_i,\pi_i)\}$, if it is feasible, then the following condition holds:
\begin{equation}
t_i>t_j \Leftrightarrow \pi_i>\pi_j.\notag
\end{equation}
\end{lemma}


Lemma \ref{lm:lemma2} shows that a longer period must be assigned with a higher price. If there is a service with longer period and a lower price, then everyone will select this data plan, and the data plans with shorter periods are meaningless.
\rev{Together with the observations from Lemma \ref{lm:lemma1}, we can find that the data plan with higher price will be assigned to the consumer with larger consumer type (i.e., standard deviation).}

From the above two lemmas and IP property, we have the following theorem, which shows the necessary and sufficient conditions for a feasible contract.

\begin{theorem}[Necessary and Sufficient conditions for a feasible contract] \label{th:theorem1}
For any contract $\mathscr{C}_d =\{(t_i,\pi_i)\}$, its IC and IR constraints are equivalent to the following conditions:
\begin{align}
\bullet~ &0 \leq t_1 \leq t_2 \leq \cdots \leq t_I.\label{eq:t11}\\
\bullet~ &\pi_I \leq V(\sigma_I,t_I).\label{eq:t12}\\
\bullet~&\pi_{i} \geq \pi_{i+1} + V(\sigma_{i+1},t_i)- V(\sigma_{i+1},t_{i+1}).\label{eq:t13}\\
\bullet~&\pi_{i} \leq \pi_{i+1} + V(\sigma_{i},t_i)- V(\sigma_{i},t_{i+1}).\label{eq:t14}
\end{align}
\end{theorem}

The feasible regions of price assignments are then
\begin{align}
&\pi_i \in [\pi_{i+1}+V(\sigma_{i+1},t_i)- V(\sigma_{i+1},t_{i+1}), \notag\\
&~~~~~~~~~~~~~~~~~~~~~~~~~~~~~\pi_{i+1} + V(\sigma_{i},t_i)- V(\sigma_{i},t_{i+1})].\notag
\end{align}
From IP property, we have $V(\sigma_{i+1},t_i)- V(\sigma_{i+1},t_{i+1}) < V(\sigma_{i},t_i)- V(\sigma_{i},t_{i+1})$. Therefore, the feasible regions of price assignments are not empty.

Then, the IC and IR constraints in Problem \ref{pro:problem1} can be substituted by the conditions shown in Theorem \ref{th:theorem1}.

\subsection{Optimality}
To solve the contract optimization problem, we first solve the optimal price assignments given the fixed period assignments, and then solve the optimal period assignments by substituting the derived price assignments. From the conditions in Theorem \ref{th:theorem1}, we can get the following lemma, which leads to the optimal price assignments.

\begin{lemma}\label{lm:lemma3}
For any feasible contract $\mathscr{C}_d =\{(t_i,\pi_i)\}$ with fixed periods $t_1 \leq t_2\leq \ldots \leq t_I$, the set of optimal price assignments $\{\bar{\pi}_i\}$ that maximizes $\sum_{i = 1}^{I}N_i(\pi_i - C(t_i))$ under the conditions in Theorem \ref{th:theorem1} is given by:
\begin{subequations}
\label{eq:lemma3}
\begin{align}
&\bar{\pi}_I = V(\sigma_I,t_I).\\
&\bar{\pi}_i = \bar{\pi}_{i+1}\! +\!V(\sigma_i,t_i)\!-\!V(\sigma_i,t_{i+1}), \forall i \in \{1,\ldots,I\!-\!1\}.
\end{align}
\end{subequations}
\end{lemma}

\begin{proof} We can observe that the price assignments in Lemma \ref{lm:lemma3} satisfy the conditions in Theorem \ref{th:theorem1}.

Since the period assignments are fixed, the total cost of the SP $\sum_{i = 1}^{I}N_iC(t_i)$ is fixed. Therefore, if there is another set of price assignments $\{\hat{\pi}_{i}\}$ that leads to a larger profit (i.e., $\sum_{i = 1}^{I}N_i\hat{\pi}_{i} > \sum_{i = 1}^IN_i\bar{\pi}_i$), then there is at least one price $\hat{\pi}_j > \bar{\pi}_j$.
According to Theorem \ref{th:theorem1}, to guarantee the feasibility of the contract, the following constraint on $\{\hat{\pi}_i\}$ must be satisfied:
\begin{equation}
\hat{\pi}_{j+1}+V(\sigma_j,t_j) - V(\sigma_j,t_{j+1}) \geq \hat{\pi}_j \label{eq:l4_2}.
\end{equation}
From \eqref{eq:lemma3} we have
\begin{equation}
\hat{\pi}_j >\bar{\pi}_j = \bar{\pi}_{j+1} + V(\sigma_j,t_j) - V(\sigma_j,t_{j+1}).\label{eq:l4_1}
\end{equation}
By substituting \eqref{eq:l4_1} into \eqref{eq:l4_2}, we have $\hat{\pi}_{j+1} > \bar{\pi}_{j+1}$, which implies $\hat{\pi}_{j+2} > \bar{\pi}_{j+2} \Rightarrow \ldots \Rightarrow \hat{\pi}_{I} > \bar{\pi}_{I}$. Since $V(\sigma_I,t_I) = \bar{\pi}_I < \hat{\pi}_I$, the IR condition is violated. Therefore, there does not exist any set of feasible price assignments $\{\hat{\pi}_i\}$ with a larger profit than $\{\bar{\pi}_i\}$.
\end{proof}

From Lemma \ref{lm:lemma3}, we can find that for fixed period assignments, the optimal price assignments are:
\begin{equation}
\bar{\pi}_i\! =\! V(\sigma_I,t_I)\! +\! \sum_{n=i}^{I-1}\Big(V(\sigma_n,t_n)\!-\!V(\sigma_n,t_{n+1})\Big), \forall i \in \mathcal{I}. \label{eq:priceass}
\end{equation}

The maximum overall profit is obtained by solving the following optimization problem
\begin{equation}
\max_{\substack{\{t_i\}}} ~\bar{R}(\{t_i\}),~~~\text{s.t.}~~0 \leq t_1 \leq t_2 \leq \cdots \leq t_I, \label{eq:optunlimited}
\end{equation}
where $\bar{R}(\{t_i\})$ is the overall profit of the optimal contract with fixed period assignments $\{t_i\}$. By subsituting the derived optimal price assignments \eqref{eq:priceass} into \eqref{eq:overallre}, we have:
\begin{align}
\bar{R}(\{t_i\}) \!&=\!\! \sum_{i=1}^I N_i\Big(\bar{\pi}_{i}-C(t_i)\Big)\notag\\
&=\!\!\sum_{i=1}^{I}\!\Big(N_iV(\sigma_i,t_i)  - N_iC(t_i)+ A_i\!\sum_{n=1}^{i-1} N_n\Big), \label{eq:overallofunlimited}
\end{align}
where $A_i = V(\sigma_i,t_i) - V(\sigma_{i-1},t_i)$ and $A_1=0$.

We define $P_i$ as $N_iV(\sigma_i,t_i)  - N_iC(t_i)+ A_i\sum_{n=1}^{i-1} N_n$ and find that $P_i$ is only based on the period $t_i$, which is designed for type $\sigma_i$ consumers. Therefore, the contract optimization problem can be divided into the following $I$ optimization problems.
\begin{equation}
\max_{\substack{t_i}} ~P_i, \forall i \in \{1,2,\ldots I\}.\label{p3}
\end{equation}
We use $\hat{t}_i$ to indicate the period that maximizes $P_i$, i.e., $\hat{t}_i = \arg \max_{\substack{t_i}}P_i$.
Since $P_i$ is concave for all $i$\footnote{It is because both $\partial^2A_i/\partial t_i^2$ and $V_{t_it_i}(\sigma_i,t_i)-C_{t_it_i}(t_i)$ are negative.
}, the optimal $\hat{t}_i$ is either \rtwo{of} the boundary points or the critical point (the point satisfying $\partial P_i/\partial t_i = 0$ and ${\partial}^2P_i/\partial {t_i}^2 \leq 0$).

We can show that if the period assignments $\{\hat{t}_k\}$ are in increasing order, then they are the optimal solution of problem \eqref{eq:optunlimited}. However, it is possible that $\{\hat{t}_i\}$ are not in increasing order, which means that they may not be feasible.
Each set of infeasible period assignments must have at least one infeasible sub-sequence, which is defined in the following definition:
\begin{definition}\label{def:definition3}
A sub-sequence $\{t_i, t_{i+1}, \ldots, t_{j}\}$ is an infeasible sub-sequence if it satisfies the following two conditions:
\begin{align}
&\bullet t_i \geq t_{i+1} \geq \ldots \geq t_j,\notag\\
&\bullet t_i > t_j.\notag
\end{align}
\end{definition}

Next, we design a mechanism to replace each infeasible sub-sequence by a feasible sub-sequence. We apply the following proposition to design the mechanism.

\begin{proposition}\label{pop:proposition2}
There are K concave functions $Y_k(y_k)$ and $\hat{y}_k = \arg\max_{\substack{y_k}}Y_k(y_k)$. If $\hat{y}_1 \geq \hat{y}_2 \geq \ldots \geq \hat{y}_K$, then the optimal solution
\begin{equation}
\{\bar{y}_k\} = \arg\max_{\substack{\{y_k\}}}\sum_{k=1}^KY_k(y_k),~~\text{s.t.} ~y_1 \leq y_2\leq \ldots \leq y_K \label{eq:algorithm1}
\end{equation}
satisfies $\bar{y}_1 = \bar{y}_2 = \ldots = \bar{y}_K$.
\end{proposition}
The proposition is proved in \cite{5738226}.

Based on Proposition \ref{pop:proposition2}, we can see that Algorithm \ref{alg1}, which is an iterative algorithm, can be used to adjust infeasible sub-sequences in $\{\hat{t}_i\}$ into feasible sub-sequences. The details of Algorithm \ref{alg1} are shown as follows.


\begin{algorithm}[t]
\caption{Iterative Algorithm to deal with infeasible sub-sequences}
\label{alg1}
\begin{algorithmic}[1]
\STATE{Initialization
$\bar{t}_i = \hat{t}_i$ for all $i \in \{1,2,\cdots, I\}$.
}
\REPEAT
\STATE{Find an infeasible sub-sequence $\{\bar{t}_m,\bar{t}_{m+1},\ldots,\bar{t}_{n}\}$.}
\STATE{Let
$\bar{t}_i = \arg\max_{\substack{t}}\sum_{k = m}^nP_k(t)$, for all $k \in \{m,m+1,\cdots,n\}$.
}
\UNTIL{$\{\bar{t}_i\}$ are feasible.}
\end{algorithmic}
\end{algorithm}


\section{Continuous-Consumer-Type with Group Division}
\label{sec:groupdivision}
In general, since consumers are mutually independent, the probability that each two consumers have the same standard deviation (i.e., $\sigma$) approaches to zero.
Therefore, it is more realistic to assume that the consumer types follow a continuous distribution.
Under such an assumption, to design a contract item for each consumer type is equivalent to providing infinite contract items, which is not realistic and not consumer friendly.
Thus, we propose a novel contract design mechanism for continuous-consumer-type model.
The mechanism divides the consumers into limited number of groups according to their types and give a contract item for each group.
Specifically, in this mechanism, we optimize the group boundaries as well as the period and price assignments in order to maximize the SP's overall profit.

\subsection{Contract Formulation}

We assume that the SP divides the consumers into $K$ groups.
Instead of designing a distinct contract item for each consumer type, the SP offers a single contract item for each group of consumer types.
We denote the set of group indices as $\mathcal{K} = \{1, 2, \ldots,K\}$, where
the minimum consumer type in the $k^{th}$ ($k \in \mathcal{K}$) group is denoted as $\sigma_k^{[min]}$ and the maximum consumer type in the $k^{th}$ group is denoted as $\sigma_k^{[max]}$.
We assume that the consumer type $\sigma$ follows a continuous distribution and the probability density function is $g(\sigma)$.
We use $\sigma_{min}$ (where $\sigma_{min}\geq 0$) and $\sigma_{max}$ (where $\sigma_{max} \geq \sigma_{min}$) to denote the minimum and maximum value of the feasible interval, i.e., $g(\sigma)>0$ only when $\sigma \in [\sigma_{min}, \sigma_{max}]$.\footnote{
To better illustrate the insights, we assume that $g(\sigma)>0$ for all $\sigma \in [\sigma_{min},\sigma_{max}]$ in this paper. For the case that $g(\sigma)=0$ for some $\sigma$ in the feasible interval, we can also show that our following analysis is valid.
}
Without loss of generality, we assume $\sigma_{min} \leq \sigma_1^{[min]}\leq\sigma_1^{[max]}\leq\sigma_2^{[min]}\leq\sigma_2^{[max]}\leq \ldots \leq\sigma_K^{[min]}\leq\sigma_K^{[max]}\leq \sigma_{max}$.
Since the consumer types follow a continuous distribution, we have $\sigma_{k-1}^{[max]}=\sigma_{k}^{[min]}$ for all $k \in \{2,3,\ldots,K\}$. The set of variables $\{\sigma_k^{[max]}|~\forall k \in \mathcal{K}\}$ is named as group boundaries.
In this paper, we let $\sigma_1^{[min]} = \sigma_{min}$ for simplicity.
We use $N$ to denote the total number of consumers, and denote the number of consumers in the $k^{th}$ group as $N_k$, where
\begin{align}
N_k \!=\! N\!\!\int_{\sigma_{k-1}^{[max]}}^{\sigma_k^{[max]}}g(x)\ud{x} = N\big(G(\sigma_k^{[max]})-G(\sigma_{k-1}^{[max]})\big). \label{eq:numberingroup}
\end{align}
In \eqref{eq:numberingroup}, $G(\sigma)$ is the cumulative distribution function of consumer type $\sigma$, and
\begin{equation}
G(\sigma_k^{[max]}) = \frac{\sum_{s=1}^kN_s}{N}.\notag
\end{equation}
In other words, $NG(\sigma)$ denotes the number of consumers with type less than or equal to $\sigma$.

Similar to our discussion on the contract for discrete-consumer-type model in Sec. \ref{sec:systemmodel} and \ref{sec:contractfando}, for each consumer group $k \in \mathcal{K}$, the SP will assign a combination including a period $t_k$ and a unit period price $\pi_k$.
We denote the contract for continuous-consumer-type model as $\mathscr{C}_c =\{(t_k,\pi_k)\}$.
The expected profit of the SP can be written as:
\begin{align}
R &= \sum_{k \in \mathcal{K}}N_k (\pi_k - C(t_k)) \label{eq:overallrev}\\
&=\sum_{k \in \mathcal{K}}N\big(G(\sigma_k^{[max]})-G(\sigma_{k-1}^{[max]})\big)(\pi_k-C(t_k)).\notag
\end{align}

To guarantee the feasibility of the contract, it should satisfy the IC and IR constraints: 1) For any consumer in group $k$, he prefers the contract item that with period $t_k$ at the price $\pi_k$ than any other contract items; 2) The SP should guarantee that the contract item designed for any consumer group leads to non-negative utility for each consumer in this group so that the consumers are willing to accept the contract designed for them. Specifically, we define,
\begin{definition}\label{def:deficct}
IC constraint:
\begin{equation}
V(\sigma,t_k)-\pi_k \!\geq\! V(\sigma,t_{k'})-\pi_{k'}, ~\forall \sigma \in [\sigma_{k-1}^{[max]}\!\!, \sigma_k^{[max]}], ~k' \neq k.\notag
\end{equation}
\end{definition}
\begin{definition}\label{def:defirct}
IR constraint:
\begin{equation}
V(\sigma,t_k) - \pi_k \geq 0,~~\forall \sigma \in [\sigma_{k-1}^{[max]}, \sigma_k^{[max]}].\notag
\end{equation}
\end{definition}
\noindent Then, the SP's profit maximization problem becomes finding the optimal group boundaries, period assignments and price assignments, i.e.,
\begin{align}
\{\sigma_k^{[max]}, (t_k, \pi_k)\} = \arg \max_{\substack{\{\sigma_k^{[max]}\} \{t_k\}  \{\pi_k\}}} R \label{eq:maxi3}
\end{align}
subject to the IC and IR constraints in Definition \ref{def:deficct}, \ref{def:defirct} and the following boundary condition:
\begin{equation}
\sigma_{min}\leq \sigma_1^{[max]}\leq\sigma_2^{[max]}\leq \cdots \leq \sigma_K^{[max]}\leq \sigma_{max}.\label{eq:thrcond}
\end{equation}

First, we can find that IP property in Proposition \ref{pop:proposition1} is still satisfied in continuous-consumer-type model.
Then, we try to find the necessary and sufficient conditions for the IC and IR constraints.
We show the first necessary condition in the following corollary.
\begin{corollary}\label{cor:cor1}
For any contract $\mathscr{C}_c =\{(t_k,\pi_k) |~ \forall k \in \mathcal{K}\}$, if it is feasible, then the following condition holds:
\begin{equation}
k>k' \Rightarrow t_k \geq t_{k'}.\notag
\end{equation}
\end{corollary}

\noindent Corollary \ref{cor:cor1} is directly obtained from Lemma \ref{lm:lemma1}.
Hence, the proof of the corollary is structurally the same as Lemma \ref{lm:lemma1} and is omitted.

The second necessary condition is shown in Lemma \ref{lm:lemma2}, which shows that a longer period must be assigned with a higher price.

From Lemma \ref{lm:lemma2}, Corollary \ref{cor:cor1} and IP property, we have the following theorem, which shows the necessary and sufficient conditions of the IC and IR constraints for the contract for continuous-consumer-type model with group division.

\begin{theorem}\label{th:theorem2}
For any contract $\mathscr{C}_c =\{(t_k,\pi_k)\}$, its IC and IR constraints are equivalent to the following conditions:
\begin{align}
&\bullet~0 \leq t_1 \leq t_2 \leq \cdots \leq t_K.\label{eq:t21}\\
&\bullet~\pi_K \leq V(\sigma_K^{[max]},t_K).\label{eq:t22}\\
&\bullet~\pi_{k} = \pi_{k+1} + V(\sigma_{k}^{[max]},t_k)- V(\sigma_{k}^{[max]},t_{k+1}).\label{eq:t234}
\end{align}
\end{theorem}


The contract optimization problem is then to maximize the SP's overall profit $R$ under Theorem \ref{th:theorem2} and boundary condition \eqref{eq:thrcond}.
In this problem, the variables needed to be optimized are: 1) the period and price of each contract item, i.e., $t_k$ and $\pi_k$; 2) the boundary of each group, i.e., $\sigma_{k}^{[max]}$.

\subsection{Contract Optimization}

According to Theorem \ref{th:theorem2}, we can obtain the optimal price assignments as follows:
\begin{align}
&\bar{\pi}_K \!=\! V(\sigma_K^{[max]},t_K), \notag\\
&\bar{\pi}_k  = \bar{\pi}_{k+1}\!+\!V(\sigma_k^{[max]},t_k)\!-\!V(\sigma_k^{[max]},t_{k+1})~\forall k \neq K,\notag
\end{align}
which implies
\begin{equation}
\bar{\pi}_k = V(\sigma_K^{[max]}\!,t_K)\! +\! \sum_{s=k}^{K-1}\big(V(\sigma_s^{[max]},t_s)\!-\!V(\sigma_s^{[max]},t_{s+1})\big).
\label{eq:2price}
\end{equation}
By substituting \eqref{eq:2price} into \eqref{eq:overallrev}, the overall profit of the SP can be rewritten as follows
\begin{subequations}
\label{eq:overalloflimited}
\begin{align}
R &= \sum_{k=1}^{K}N_k(\bar{\pi}_k-C(t_k))\\
& = \sum_{k=1}^KN_k\Big(V(\sigma_K^{[max]},t_K)\!-\!C(t_k)\notag\\
&~~~~~~+\!\sum_{s=k}^{K-1} \big(V(\sigma_s^{[max]},t_s)
-V(\sigma_s^{[max]},t_{s+1})\big)\Big) \label{eq:step0}\\
&=\sum_{k=1}^{K}\Big(N_kV(\sigma_k^{[max]},t_k)\! -\! N_kC(t_k)\!+\! A_k\!\sum_{s=1}^{k-1} N_s\Big)\label{eq:step2}\\
&=\sum_{k=1}^{K-1} NG(\sigma_k^{[max]})\big(V(\sigma_k^{[max]}, t_k)-V(\sigma_k^{[max]},t_{k+1})\big)\notag\\
&~~~+\sum_{k=1}^{K-1}NG(\sigma_k^{[max]})\big(C(t_{k+1})-C(t_k)\big)\notag\\
&~~~+NG(\sigma_K^{[max]})\big(V(\sigma_K^{[max]},t_K)-C(t_K)\big),\label{eq:step1}
\end{align}
\end{subequations}
where $A_k = V(\sigma_k^{[max]},t_k) - V(\sigma_{k-1}^{[max]},t_k)$ and $A_1=0$.

\setcounter{equation}{31}
\begin{figure*}[ht]
\normalsize
\begin{equation}
\frac{\partial Q_k(\sigma_k)}{\partial \sigma_k}=Ng(\sigma_k)\Big(V(\sigma_k,t_k)-V(\sigma_k,t_{k+1})
+\frac{G(\sigma_k)\big(V_{\sigma_k}(\sigma_k,t_k)-V_{\sigma_k}(\sigma_k,t_{k+1})\big)}{g(\sigma_k)}+C(t_{k+1})\!-\!C(t_k)\Big). \label{eq:firstorder}
\end{equation}
\end{figure*}

\setcounter{equation}{35}
\begin{figure*}[ht]
\normalsize
\begin{equation}
\frac{H_k(\sigma_k)}{\sigma_k} \!=\!\!\! \int_{t_{k}}^{t_{k+1}}\!
\frac{\alpha e^{-\frac{t\Delta q^2}{2\sigma_k^2}}(t\Delta q^2+\sigma_k^2)}
{2\sqrt{2\pi t}\sigma_k^2 g(\sigma_k)t}
\Big(
\frac{t\Delta q^2 (\sigma_k^2-t\Delta q^2)}{\sigma_k^3(\sigma_k^2+t\Delta q^2)}G(\sigma_k)
 -\frac{2g^2(\sigma_k)-g_{\sigma_k}(\sigma_k)G(\sigma_k)}{g(\sigma_k)}
\Big)\ud{t}.\label{eq:firsth}
\end{equation}
\end{figure*}

\setcounter{equation}{24}

\subsubsection{Introduction of The Alternative Maximizing Algorithm}
Finding the optimal period assignments $\{t_k\}$ and group boundaries $\{\sigma_k^{[max]}\}$ that can maximize the SP's overall profit with price assignments in \eqref{eq:2price} is very challenging because problem \eqref{eq:maxi3} is NP-hard \cite{li2014dynamic}.
Therefore, we introduce an alternative maximizing algorithm to find a sub-optimal solution.
In this algorithm, we divide the variables into two groups, where the first group contains all the period assignments $\{t_k\}$, and the other group contains all the group boundaries $\{\sigma_k^{[max]}\}$.
At the beginning of the algorithm, we divide the consumers into $K$ groups by randomly generating $K$ group boundaries $\{\sigma_k^{[max]}\}$ such that $\sigma_{min} < \sigma_1^{[max]} < \sigma_2^{[max]} < \cdots < \sigma_K^{[max]}<\sigma_{max}$.
Then, we iterate the following two steps.
In the first step, we keep the group boundaries $\{\sigma_k^{[max]}\}$ unchanged and maximize the overall profit $R$ by tuning period assignments $\{t_k\}$. In the second step, we keep period assignments $\{t_k\}$ (which are obtained by solving the problem in the previous step) unchanged and update group boundaries $\{\sigma_k^{[max]}\}$ to maximize $R$. For the rest of the paper, we will simply use $\sigma_k$ to denote $\sigma_k^{[max]}$ which is the threshold between group $k$ and group $k+1$.

The details of these two steps are as follows.

\begin{itemize}
\item \textbf{Step I}:
In this step, we find the optimal period assignments that can maximize the overall profit $R$ with fixed group boundaries $\{\sigma_k\}$.
Specifically, we have the following problem
\begin{equation}
\max_{\substack{\{t_i\}}}~ R,~~~\text{s.t.}~~0\leq t_1\leq t_2 \leq \cdots \leq t_K.\label{eq:optlimited}
\end{equation}

We can show that the overall profit in \eqref{eq:step2} is structurally similar to the overall profit \eqref{eq:overallofunlimited} in Sec. \ref{sec:contractfando}.
Therefore, this optimization problem can be solved through the same method that solves problem \eqref{eq:optunlimited}.

For writing convenience, we define
\begin{equation}
P_k=N_kV(\sigma_k,t_k) - N_kC(t_k)+ A_k\sum_{s=1}^{k-1}N_s, \label{eq:pc}
\end{equation}
where $A_k = V(\sigma_k,t_k) - V(\sigma_{k-1},t_k)$ and $A_1 = 0$.
Then, we divide the problem \eqref{eq:optlimited} into $K$ optimization problems
\begin{equation}
\max_{\substack{t_k}} P_k, ~~~\forall k \in \mathcal{K}. \label{eq:ppplimited}
\end{equation}
Since $P_k$ is concave for all $k$, the period assignment $\hat{t}_k$ that can maximize $P_k$ is at the boundary points or at the critical point, i.e.,
\begin{equation}
\label{eq:opt_limi}
\hat{t}_k = \left\{ \begin{array}{ll}
0 & \textrm{if $\tilde{t}_k<0$},\\
\tilde{t}_k & \textrm{if $\tilde{t}_k \geq 0$},
\end{array} \right.
\end{equation}
where $\tilde{t}_k$ is the solution of $\partial P_k/\partial t_k = 0$.

By denoting the optimal solution of problem \eqref{eq:optlimited} as $\{\bar{t}_k\}$, we can see that if the period assignments $\{\hat{t}_k\}$ from \eqref{eq:ppplimited} are in increasing order, then $\bar{t}_k = \hat{t}_k$ for all $k$.
However, if $\{\hat{t}_k\}$ are not in increasing order, which means that they may not be feasible, we need to use Algorithm \ref{alg1} to adjust infeasible period assignments to make them feasible.
In the input of the algorithm, we define $P_k$ as in \eqref{eq:pc} and let $I = K$.

\item \textbf{Step II}:
In this step, we find the optimal group boundaries that can maximize the overall profit $R$ with fixed period assignments $\{t_k\}$.
Specifically, we have the following optimization problem:
\begin{equation}
\max_{\substack{\{\sigma_k\}}}\ R,
~~\text{s.t.}~\sigma_{min}\leq\sigma_{1}\leq \cdots \leq\sigma_K \leq \sigma_{max}. \label{eq:problemstep1}
\end{equation}

By defining $Q_k(\sigma_k)$ as
\begin{subequations}
\label{eq:Q}
\begin{align}
&Q_{k}(\sigma_k)\! = \! NG(\sigma_k)\big(V(\sigma_k, t_k)\!-\!V(\sigma_k,t_{k+1})\notag\\
&~~~~~~~~~~~~~+\!C(t_{k+1})\!-\!C(t_k)\big),~\forall k \in \{1,2,\cdots,K\!-\!1\},\label{eq:qk}\\
&Q_{K}(\sigma_K) \!\!= \!N G(\sigma_K)\big(V(\sigma_{K},t_K)\!-\!C(t_K)\big),\label{eq:qkk}
\end{align}
\end{subequations}
we can find that the overall revenue $R$ in \eqref{eq:step1} can be represented as the summation of $Q_{k}(\sigma_k)$.
\eqref{eq:Q} implies that $Q_k(\sigma_k)$ is only related to $\sigma_k$, i.e., the group boundary between group $k$ and group $k+1$, and independent of the other group boundaries $\{\sigma_s |~ \forall s \in \mathcal{K},~ s \neq k\}$.
Therefore, the best group boundaries for \eqref{eq:problemstep1}, denoted by $\{\bar{\sigma}_k\}$, can be computed
by separately maximizing each of $Q_k(\sigma_k)$, $\forall k \in \mathcal{K}$.

We use $\hat{\sigma}_k$ to denote the group boundary that maximizes $Q_k(\sigma_k)$, i.e.,
\begin{equation}
\hat{\sigma}_k = \arg\max_{\substack{\sigma_k}}\ Q_k(\sigma_k), \forall k \in \mathcal{K}. \label{eq:hatsigma}
\end{equation}
If the group boundaries $\{\hat{\sigma}_k\}$ obtained by solving \eqref{eq:hatsigma} are in increasing order, $\{\hat{\sigma}_{k}\}$ are exactly the solution of \eqref{eq:problemstep1}, i.e., $\bar{\sigma}_k=\hat{\sigma}_k$ $\forall k \in \mathcal{K}$. If $\{\hat{\sigma}_k\}$ are not in increasing order, some further steps are needed to obtain the optimal group boundaries of \eqref{eq:problemstep1} from $\{\hat{\sigma}_k\}$.

In order to solve the problem \eqref{eq:hatsigma}, we find the first order derivative of $Q_k(\sigma_k)$ \rtwo{with} respect to $\sigma_k$, which is shown in \eqref{eq:firstorder} on the top of this page.
Although the form of $\partial Q_k(\sigma_k)/\partial \sigma_k$ is complicated, we obtain the unimodality of $Q_k(\sigma_k)$ in the following Theorem.
\setcounter{equation}{32}

\begin{theorem}\label{th:theorem3}
If the distribution of the consumer types $\sigma$ satisfies the condition
\begin{equation}
\frac{2g^2(\sigma)-g_{\sigma}(\sigma)G(\sigma)}{g(\sigma)} \geq \left\{ \begin{array}{ll}
0 & \textrm{if $\sigma=0$},\\
\frac{3-2\sqrt{2}}{\sigma}G(\sigma) & \textrm{if $\sigma > 0$},
\end{array} \right. \label{eq:distribution}
\end{equation}
\noindent then with the fixed period assignments, the formula $Q_k(\sigma_k)$ is unimodal with respect to $\sigma_k$ for all $k \in \mathcal{K}$.
\end{theorem}

\begin{proof}
By defining $H_k(\sigma_k)$ as
\begin{align}
H_k(\sigma_k) \!&= \!V(\sigma_k,t_k)-V(\sigma_k,t_{k+1})\notag\\
&~~~~~~~~+\frac{G(\sigma_k)\big(V_{\sigma_k}(\sigma_k,t_k)-V_{\sigma_k}(\sigma_k,t_{k+1})\big)}{g(\sigma_k)}\notag\\
&=\!\!\int_{t_{k+1}}^{t_k}\!\!V_t(\sigma_k,t)\!+\!\frac{G(\sigma_k)}{g(\sigma_k)}V_{\sigma_k,t}(\sigma_{k},t)\ud{t},\label{eq:H}
\end{align}
the first order derivation of $Q_k(\sigma_k)$ \rtwo{with} respect to $\sigma_k$ can be rewritten as:
\begin{equation}
\frac{\partial Q_k(\sigma_k)}{\partial \sigma_k} = Ng(\sigma_k)(H_k(\sigma_k)+C(t_{k+1})-C(t_{k})).
\end{equation}
Since $g(\sigma_k)$ is positive and $C(t_{k+1})-C(t_k)$ is a constant for fixed period assignments, $Q_k(\sigma_k)$ is unimodal if $H_k(\sigma_k)$ is monotonic \rtwo{with} respect to $\sigma_k$.

To study the monotonicity of $H_k(\sigma_k)$, we need to find the first order derivative of it \rtwo{with respect} to $\sigma_k$, which is shown in \eqref{eq:firsth} on the top of this page.
Before we find the sign of the first order derivation, we first show the following lemma.
\setcounter{equation}{36}
\begin{lemma}\label{lm:lemma5}
For any $t \geq 0$ and $\sigma > 0$, we have
\begin{equation}
\frac{t\Delta q^2(\sigma^2-t\Delta q^2)}{\sigma^3(\sigma^2+t\Delta q^2)} \leq \frac{3-2\sqrt{2}}{\sigma}.\label{eq:max}
\end{equation}
\end{lemma}


We can show that $\frac{t\Delta q^2(\sigma^2-t\Delta q^2)}{\sigma^3(\sigma^2+t\Delta q^2)}G(\sigma) = 0$ when $\sigma=0$, and $\frac{t\Delta q^2(\sigma^2-t\Delta q^2)}{\sigma^3(\sigma^2+t\Delta q^2)}G(\sigma) \leq \frac{3-2\sqrt{2}}{\sigma}G(\sigma)$ when $\sigma>0$.
Hence, if the distribution of the consumer type satisfies \eqref{eq:distribution} in Theorem \ref{th:theorem3} for all $\sigma\geq 0$, then $\partial H_k(\sigma_k)/ \partial \sigma_k$ is always non-positive and $H_k(\sigma_k)$ crosses zero at most once.
Therefore, $Q_k(\sigma_k)$ is unimodal with respect to $\sigma_k$.
\end{proof}

In next subsection, we will show that \eqref{eq:distribution} applies for some typical distributions.
Then, according to Theorem \ref{th:theorem3}, $Q_k(\sigma_k)$ is an unimodal function with respect to $\sigma_k$.
Therefore, the optimal $\hat{\sigma}_k$ is at the boundary point or the critical point, i.e.,
\begin{equation}
\label{eq:opsig_limi}
\hat{\sigma}_k = \left\{ \begin{array}{ll}
0 & \textrm{if $\tilde{\sigma}_k<0$},\\
\tilde{\sigma}_k & \textrm{if $\tilde{\sigma}_k \geq 0$},
\end{array} \right.
\end{equation}
where $\tilde{\sigma}_k$ is the solution of $\partial Q_k(\sigma_k)/\partial \sigma_k = 0$.

The group boundaries $\{\hat{\sigma}_k\}$, which are obtained by separately maximizing each of $Q_k(\sigma_k)$, may not be in increasing order, which means they may not be feasible.
Each set of infeasible group boundaries must have at least one infeasible sub-sequence, which is defined in Definition \ref{def:definition3}.
In order to adjust an infeasible sequence $\{\hat{\sigma}_i, \hat{\sigma}_{i+1}, \cdots \hat{\sigma}_j\}$ to a feasible sub-sequence, we first show a property of $Q_k(\sigma_k)$ in the following proposition.

\begin{proposition}\label{pop:proposition3}
For any $i,j \in \mathcal{K}$ and $j \geq i$, the function $\sum_{k=i}^jQ_{k}(\sigma)$ is a unimodal function with respect to $\sigma$.
\end{proposition}


Then we apply the following proposition to design a mechanism to deal with the infeasible sub-sequence in $\{\sigma_k\}$.
\begin{proposition}\label{pop:proposition4}
There are $2$ unimodal functions $Q_1(\sigma_1)$ and $Q_2(\sigma_2)$. If $\hat{\sigma}_1 \geq \hat{\sigma}_2$, where $\hat{\sigma}_1 = \arg\max_{\substack{\sigma1}}Q_1(\sigma_1)$ and $\hat{\sigma}_2 = \arg\max_{\substack{\sigma_2}}Q_2(x_2)$, then the optimal solution
\begin{equation}
\{\bar{\sigma}_k\} = \arg\max_{\substack{\{\sigma_k\}}}\sum_{k=1}^2Q_k(\sigma_k),~~\text{s.t.} ~\sigma_1 \leq \sigma_2\notag
\end{equation}
satisfies $\bar{\sigma}_1 = \bar{\sigma}_2$.
\end{proposition}


Since the unimodal functions $\{Q_k(\sigma_k)\}$ have the property shown in Proposition \ref{pop:proposition3}, Proposition \ref{pop:proposition4} can be extended to a more general form: for any $j \geq i$, if $\hat{\sigma}_i \geq \hat{\sigma}_{i+1} \geq \cdots \geq \hat{\sigma}_j$, where $\hat{\sigma}_k = \arg\max_{\substack{\sigma_k}}Q_k(\sigma_k)$, then the optimal solution $\{\bar{\sigma}_k\} = \arg\max_{\substack{\sigma_k}}\sum_{k=i}^j Q_k(\sigma_k)$ subject to $\sigma_i\leq \sigma_{i+1} \leq\cdots \leq \sigma_j$ satisfies $\bar{\sigma}_i =\bar{\sigma}_{i+1}=\cdots =\bar{\sigma}_j$.

By means of Proposition \ref{pop:proposition4}, we can use an algorithm similar to Algorithm \ref{alg1}, to adjust infeasible sub-sequences in $\{\sigma_k\}$ to feasible sub-sequences.



\end{itemize}

The details of the alternative maximizing algorithm are illustrated in Algorithm \ref{alg2}.

\begin{algorithm}[t]
\caption{Alternative Maximizing Algorithm}
\label{alg2}
\begin{algorithmic}[1]
\STATE{Initialize $K$ arbitrary groups�boundaries $\{\sigma_k\}$, where $\sigma_k \geq 0$ $\forall k$ and $\sigma_{min}<\sigma_1 < \sigma_2 < \cdots < \sigma_K < \sigma_{max}$.}
\REPEAT
\STATE{\textbf{Step I}:}
\FORALL{$k \in \mathcal{K}$}
\STATE{Define $P_k$ as $N_kV(\!\sigma_k^{[max]}\!\!,t_k\!) \!-\! N_kC(t_k)\!+\! A_k\!\sum_{s=1}^{k-1}\!N_s$.
}
\STATE{Solve the problem: $\max_{\substack{t_k}} P_k$, where the optimal point $\hat{t}_k$ can be obtained from \eqref{eq:opt_limi}.}
\ENDFOR
\STATE{If the period assignments obtained are not feasible, using Algorithm \ref{alg1} to adjust infeasible sub-sequences into feasible sub-sequences (the definition of infeasible sub-sequence is in Definition \ref{def:definition3}).}

\STATE{\textbf{Step II:}}
\FORALL{$k \in \mathcal{K}$}
\STATE{Calculate $Q_k(\sigma_k)$ according to \eqref{eq:Q}.}
\STATE{Solve the problem: $\max_{\substack{\sigma_k}}Q_k(\sigma_k)$, where the optimal point $\hat{\sigma}_k$ can be obtained from \eqref{eq:opsig_limi}.}
\ENDFOR
\STATE{Initialization $\bar{\sigma}_k = \hat{\sigma}_k$ for all $k \in \mathcal{K}$.}
\REPEAT
\STATE{Find an infeasible sub-sequence $\{\bar{\sigma}_i, \bar{\sigma}_{i+1}, \cdots, \bar{\sigma}_{j}\}$.}
\STATE{Let $\bar{\sigma}_{k} = \arg\max_{\substack{\sigma}}\sum_{k=i}^jQ_k(\sigma_k)$, for all $k \in \{i,i+1,\cdots,j\}$.}
\UNTIL{$\{\bar{\sigma}_k\}$ are feasible.}
\UNTIL{Convergence.}
\end{algorithmic}
\end{algorithm}

\subsubsection{Convergence of The Alternative Maximizing Algorithm}
In the alternative maximizing algorithm, with arbitrary initialized group boundaries, we alternatively update the period assignments $\{t_k\}$ and group boundaries $\{\sigma_k\}$ in order to maximize the overall profit $R$.

In each step of the alternative maximizing algorithm, we try to adjust the period assignments or group boundaries to maximize $R$. Hence, the value of $R$ is monotonically increasing.
Since the value of overall profit $R$ is upper bounded, the algorithm will finally converge.

\subsection{Analysis of Some Typical Distributions of The Consumer Types}

In this subsection, we will show that \eqref{eq:distribution} in Theorem \ref{th:theorem3} applies for some typical distributions of consumer types, including uniform distribution, exponential distribution and truncated normal distribution.

\subsubsection{Uniform Distribution}
We first study the case that the consumer types follow a uniform distribution.
Specifically,
we have
\begin{align}
&g(\sigma) = \frac{1}{\sigma_{max}-\sigma_{min}},\label{eq:g}\\
&g_{\sigma}(\sigma) = 0, \label{eq:g_s}\\
&G(\sigma) = \frac{\sigma-\sigma_{min}}{\sigma_{max}-\sigma_{min}}.\label{eq:gg}
\end{align}
By substituting \eqref{eq:g}, \eqref{eq:g_s} and \eqref{eq:gg} into $\frac{2g^2(\sigma)-g_{\sigma}(\sigma)G(\sigma)}{g(\sigma)G(\sigma)}$ (i.e., the left hand of \eqref{eq:distribution}), we have
\begin{align}
&\frac{2g^2(\sigma)-g_{\sigma}(\sigma)G(\sigma)}{g(\sigma)G(\sigma)} = \frac{2g(\sigma)}{G(\sigma)}\notag\\
&= \frac{2}{\sigma-\sigma_{min}}= \frac{2}{\sigma}\frac{\sigma}{\sigma-\sigma_{min}} \overset{(a)}{\geq} \frac{2}{\sigma} \geq \frac{3-2\sqrt{2}}{\sigma},\notag
\end{align}
where $(a)$ is from the fact that $\frac{\sigma}{\sigma-\sigma_{min}} \geq 1$ for all $\sigma \in [\sigma_{min}, \sigma_{max}]$.
In conclusion, \eqref{eq:distribution} holds for the case of uniform distribution in consumer types.

\subsubsection{Exponential Distribution}
Next, we study the case that the consumer types follow an exponential distribution.
Specifically, we have
\begin{align}
&g(\sigma) = \lambda e^{-\lambda \sigma},\label{eq:exg}\\
&g_{\sigma}(\sigma) = -\lambda^2e^{-\lambda \sigma}, \label{eq:exg_s}\\
&G(\sigma) = 1-e^{-\lambda \sigma}. \label{eq:exgg}
\end{align}
Here, $\lambda>0$ is the rate parameter of the exponential distribution. By substituting \eqref{eq:exg}, \eqref{eq:exg_s} and \eqref{eq:exgg} into $\frac{2g^2(\sigma)-g_{\sigma}(\sigma)G(\sigma)}{g(\sigma)}$, we have
\begin{equation}
\frac{2g^2(\sigma)-g_{\sigma}(\sigma)G(\sigma)}{g(\sigma)} = \lambda (1+e^{-\lambda \sigma}).\notag
\end{equation}

When $\sigma = 0$, we have $\lambda(1+e^{-\lambda\sigma}) = 2\lambda >0$, which satisfies \eqref{eq:distribution} in Theorem \ref{th:theorem3}.

When $\sigma > 0$, we have
\begin{align}
\frac{2g^2(\sigma)-g_{\sigma}(\sigma)G(\sigma)}{g(\sigma)} &= \lambda (1+e^{-\lambda \sigma})\notag\\
 &= \frac{\lambda\sigma}{1-e^{-\lambda\sigma}}\frac{1+e^{-\lambda\sigma}}{\sigma}(1-e^{-\lambda\sigma})\notag\\
 &=\frac{x}{1-e^{-x}}\frac{1+e^{-\lambda\sigma}}{\sigma}(1-e^{-\lambda\sigma}),
\label{eq:expyuanshi}
\end{align}
where $x = \lambda \sigma$ and $x > 0$.
Before finding the lower bound of \eqref{eq:expyuanshi}, we first derive the following proposition.
\begin{proposition}
\label{pop:proposition5}
For $x>0$, the formula $\frac{x}{1-e^{-x}}$ is lower bounded by $1$.
\end{proposition}


According to Proposition \ref{pop:proposition5}, we have
\begin{align}
&\frac{2g^2(\sigma)-g_{\sigma}(\sigma)}{g(\sigma)G(\sigma)} = \frac{x}{1-e^{-x}}\frac{1+e^{-\lambda\sigma}}{\sigma}(1-e^{-\lambda\sigma})\notag\\
&\geq \frac{1+e^{-\lambda\sigma}}{\sigma}(1-e^{-\lambda\sigma})\overset{(b)}{\geq} \frac{3-2\sqrt{2}}{\sigma}G(\sigma),
\end{align}
where (b) is from the fact $1+e^{-\lambda\sigma}>1$ for any $\sigma>0$.
In conclusion, \eqref{eq:distribution} holds for the case of exponential distribution in consumer types.

\subsubsection{Truncated Normal Distribution}

In truncated normal distribution, we can not obtain the expression of the cumulative function $G(\sigma)$ in closed-form due to the non-integrability of the formula.
Hence, we are not able to show \eqref{eq:distribution} analytically.
In this case, we use numerical results to illustrate that \eqref{eq:distribution} holds for various parameters.

When $\sigma = 0$, we have $\frac{2g^2(\sigma)-g_{\sigma}(\sigma)G(\sigma)}{g(\sigma)} = 2g(\sigma) \geq 0$, which means that \eqref{eq:distribution} holds for the case $\sigma=0$.

When $\sigma >0$, to simplify the notations, we define a function $F(\sigma)$ as
\begin{equation}
F(\sigma) = \frac{2g^2(\sigma)-g_{\sigma}(\sigma)G(\sigma)}{g(\sigma)}- \frac{3-2\sqrt{2}}{\sigma}G(\sigma).\notag
\end{equation}
If we can show that $F(\sigma) \geq 0$ for various parameters, which means $\frac{2g^2(\sigma)-g_{\sigma}(\sigma)G(\sigma)}{g(\sigma)} \geq \frac{3-2\sqrt{2}}{\sigma}G(\sigma)$, then we can see that \eqref{eq:distribution} holds for truncated normal distribution with various parameters.

In our simulation settings, we let the minimum value of truncated normal distribution $a=0$ and the maximum value of truncated normal distribution $b=6$.
We use $M$ and $W$ to denote the mean and the standard deviation of the corresponding normal distribution.
\begin{figure}[tbp]
	\centering
		\includegraphics[width=85mm]{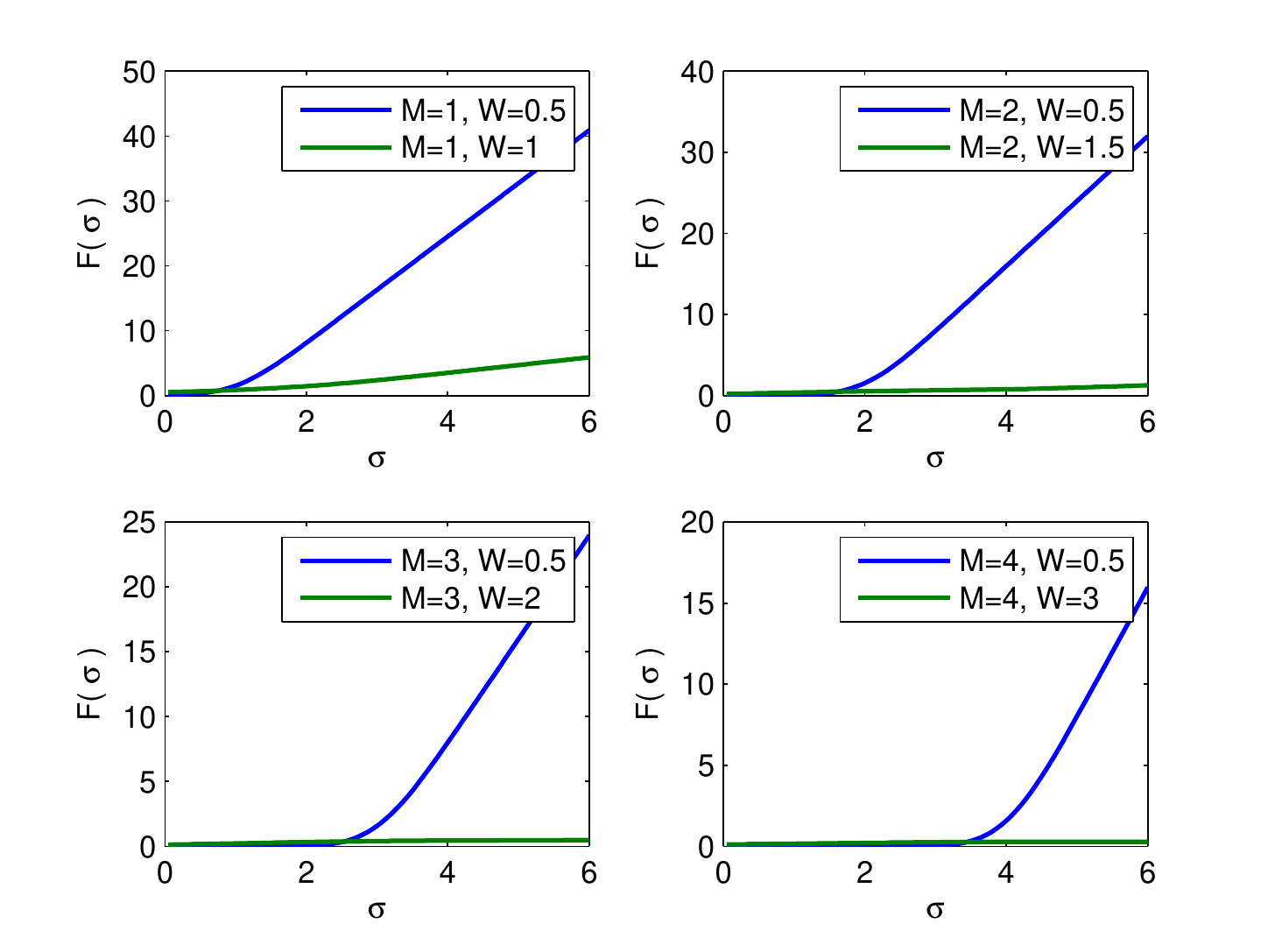}
	\caption{Examples of $F(\sigma)$ with different truncated normal distributions.}
	\label{fig:truncateddis}
\end{figure}
As shown in Figure \ref{fig:truncateddis}, the value of $F(\sigma)$ is positive for different combinations of $M$ and $W$ values.
Therefore, \eqref{eq:distribution} holds for truncated normal distribution with various parameters.




\section{Simulation Results}
\label{sec:simulation}
In the simulation, we first implement the proposed period-price contract in discrete-consumer-type model, and then implement the proposed period-price contract in continuous-consumer-type model.
Without loss of generality, we set the predefined parameter $\alpha = 1$ and the average data demand per unit period as $\mu=13$.
We assume the data cap of the unit period data plan is $q = 15$ and the cost function of the SP is $C(t) = 0.5t+10$ \footnote{
According to \cite{a8, a9, a10}, a consumer usually chooses a data plan with monthly data cap larger than his average consumption.
Period of data plan helps an SP to manage its network capacity, because an SP should make sure a corresponding network capacity is prepared during the whole period in case that the consumers consume all data quota for the whole period in a very short time.
Hence, a larger period requires the SP to prepare more network capacity and will lead to a higher cost. Here for simplicity, we consider a linear-form cost, which has been widely used to model an operator's operational cost (e.g., \cite{a11, a12}).
}.

\subsection{Discrete-Consumer-Type}

\begin{figure}[tbp]
	\centering
		\includegraphics[width=83mm]{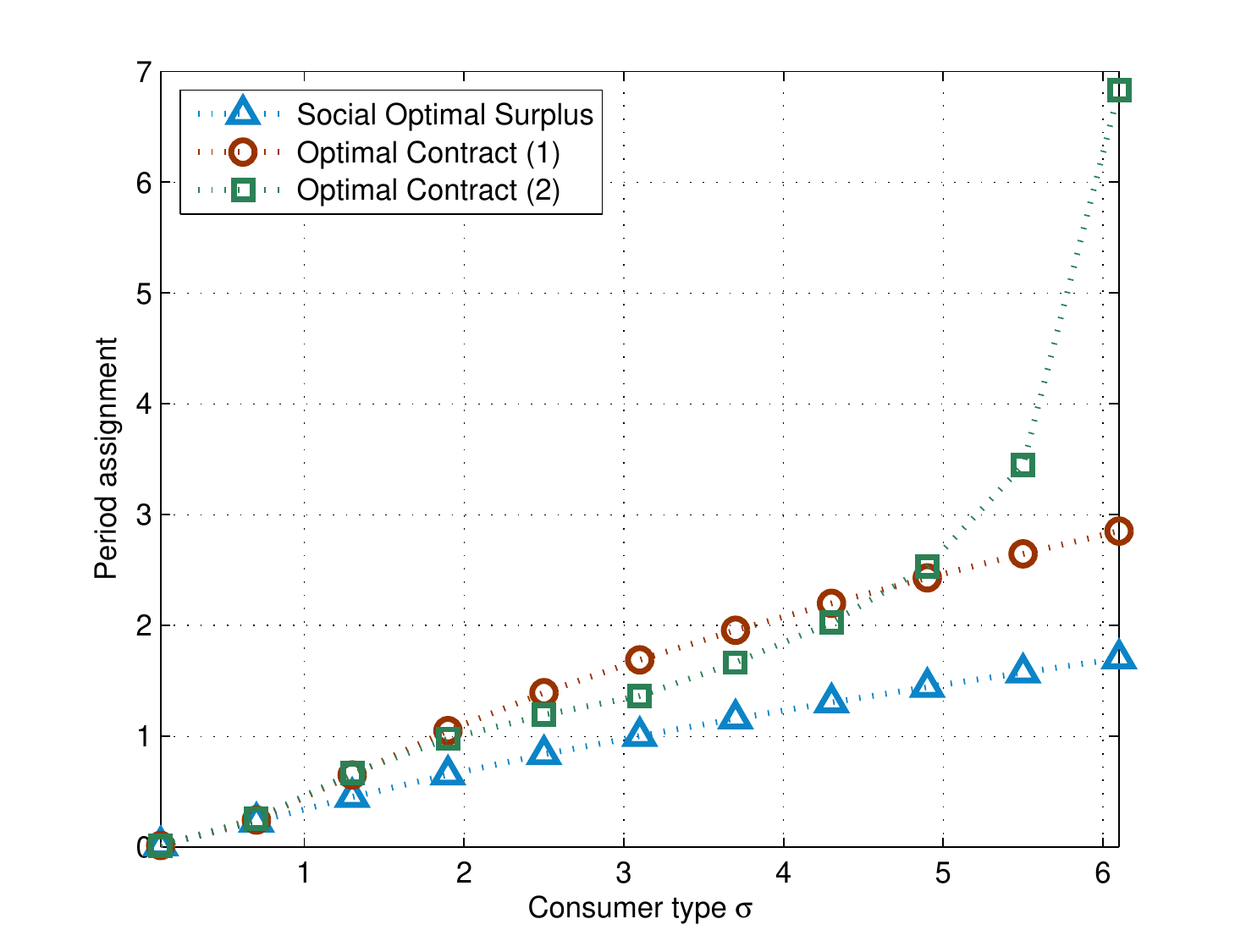}
		\vspace{-2mm}
	\caption{Period assignments for discrete-consumer-type model.}
	\label{fig:d-period}
\end{figure}

\begin{figure}[tbp]
	\centering
		\includegraphics[width=83mm]{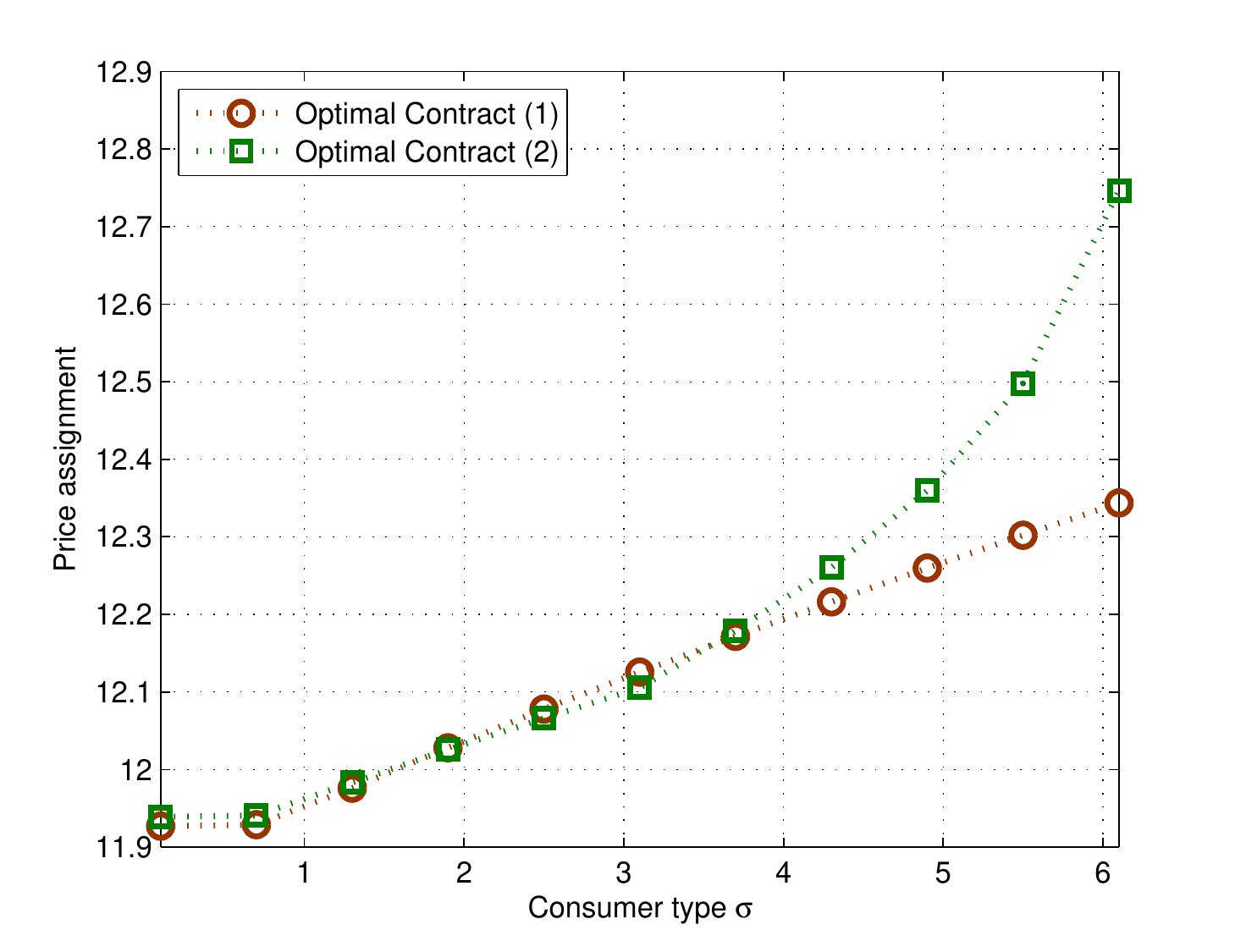}
		\vspace{-2mm}
	\caption{Price assignments for discrete-consumer-type model.}
	\label{fig:d-price}
\end{figure}

In discrete-consumer-type model, we assume the number of consumer types $I = 11$.
The set of consumer types is $\Sigma = \{0.1,0.7,1.3,\cdots,6.1\}$.

We run the simulation of the optimal contract in two cases. In Case (1), the numbers of consumers in each type are identical. In Case (2), the numbers are distributed in a mountain shape, which means the probability of medium is large.

\begin{figure}[tbp]
	\centering
		\includegraphics[width=83mm]{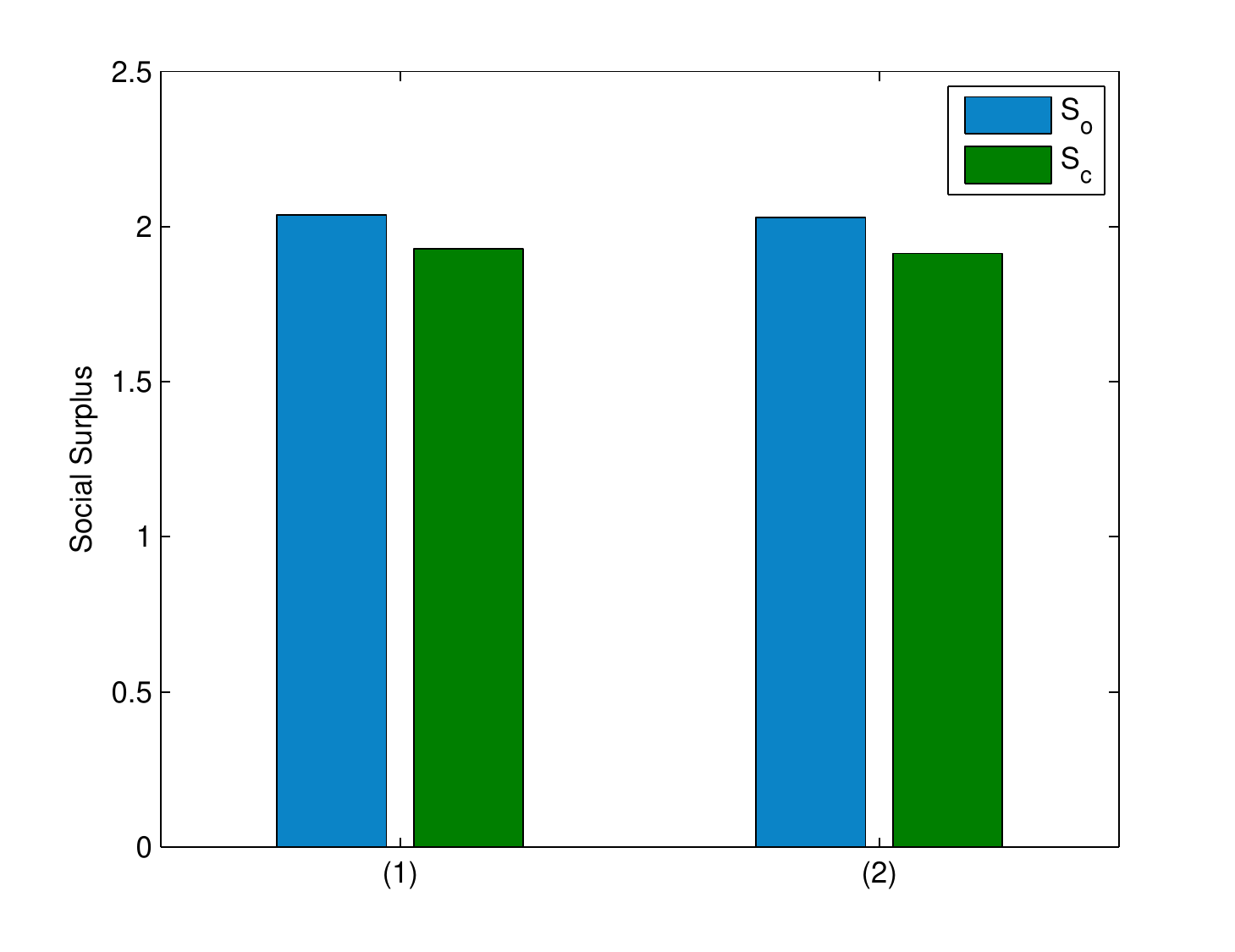}
		\vspace{-2mm}
	\caption{Comparison between optimal contract and social surplus maximization scheme.}
	\label{fig:d-social}
\end{figure}

\begin{figure}[tbp]
	\centering
		\includegraphics[width=83mm]{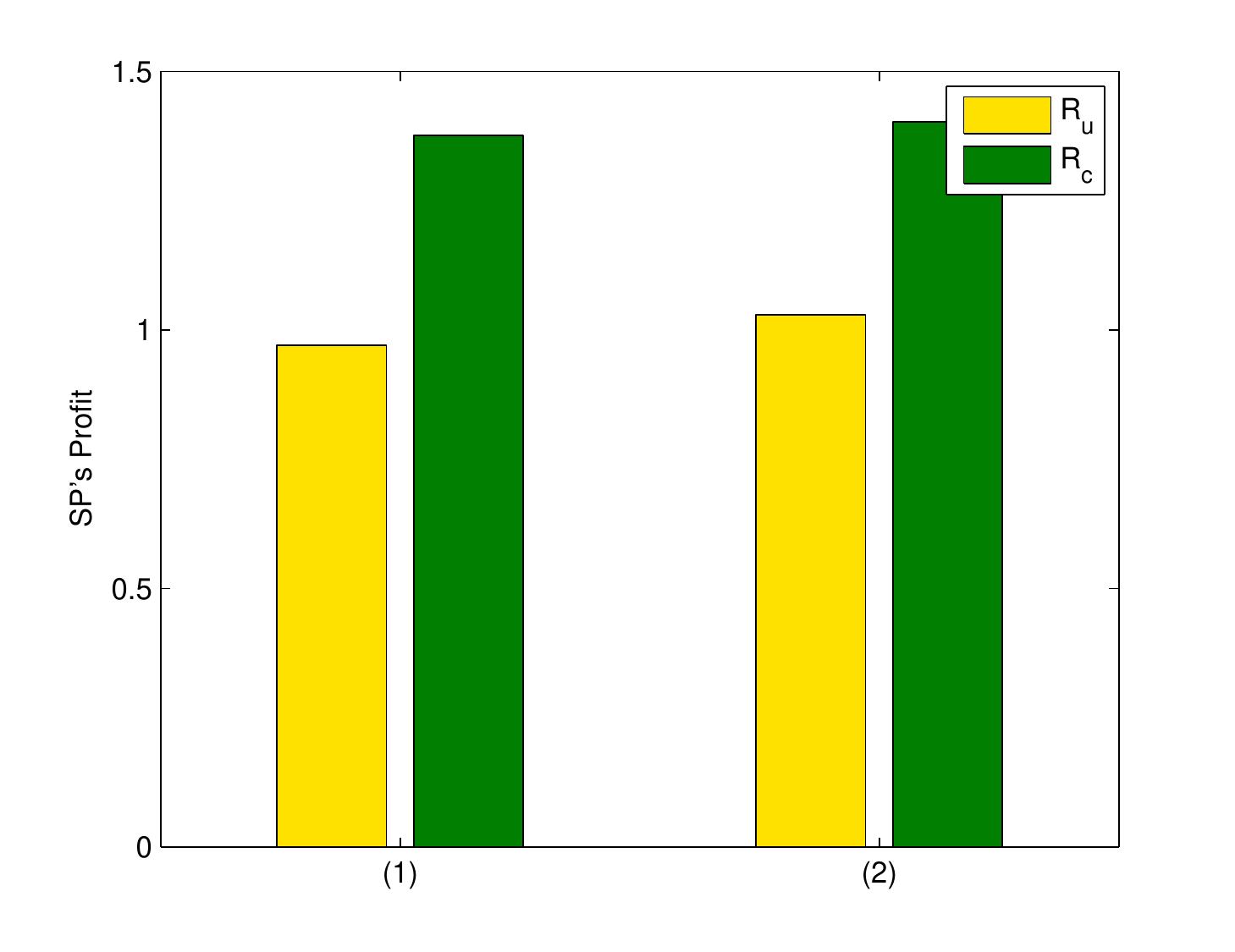}
		\vspace{-2mm}
	\caption{Comparison between optimal contract and the monthly-period scheme.}
	\label{fig:d-spprofit}
\end{figure}

We define the social surplus generated by the contract with period $t_i$, denoted by $S(\sigma_i, t_i)$, as the aggregate utilities of SP and the consumer with type $\sigma_i$, i.e.,
\begin{equation}
S(\sigma_i,t) \triangleq V(\sigma_i,t_i) - C(t_i).\notag
\end{equation}
Figure \ref{fig:d-period} and Figure \ref{fig:d-price} show the period and price assignments in the optimal contract.
The blue curve in Figure \ref{fig:d-period} presents the social optimal period assignments, which maximize the social surplus.
Specifically, we have $t_i = \arg\max_{\substack{t}}S(\sigma_i,t)$.
From Figure \ref{fig:d-period}, we can see that the social optimal period assignments are always smaller than that in the optimal contract. This is because the optimal contract is aimed to maximize the SP's overall profit rather than the social surplus. The SP prefers to increase the period assigned to the higher type consumers in order to increase the interest of the lower type consumers in the short period contract items. Hence, by increasing the price of the short period contract items, the SP can increase the profit.

Figure \ref{fig:d-social} shows the social surplus in the optimal contract and social optimal assignments. The bars $S_o$ denote the social surplus in the social optimal period assignments while the social surplus of our proposed optimal contract is represented by the bars $S_c$.
We can see that $S_o$ is larger than $S_c$ in both Case (1) and Case (2), since the optimal contract is aimed to maximize the SP's profit rather than social surplus.
However, our proposed optimal contract can achieve around $93\%$ of the maximum social surplus.

The comparison between our proposed optimal contract and conventional monthly-period scheme in terms of SP's profit is shown in Figure \ref{fig:d-spprofit}.
The bars $R_u$ and $R_c$ denote the profits of the SP in the monthly-period scheme and our optimal contract, respectively.
We can see that our optimal contract can increase the SP's profit by 41\% and 37\% for Case (1) and Case (2) correspondingly.
\wei{This is because the optimal contract increases the period assigned to the consumers with larger consumer types in order to increase the interest of the smaller type consumers in the short period contract items.
Hence, by increasing the price of the long
period contract items and decreasing the period of the short period items, the SP can increase its profit.}

\begin{figure}[tbp]
	\centering
		\includegraphics[width=83mm]{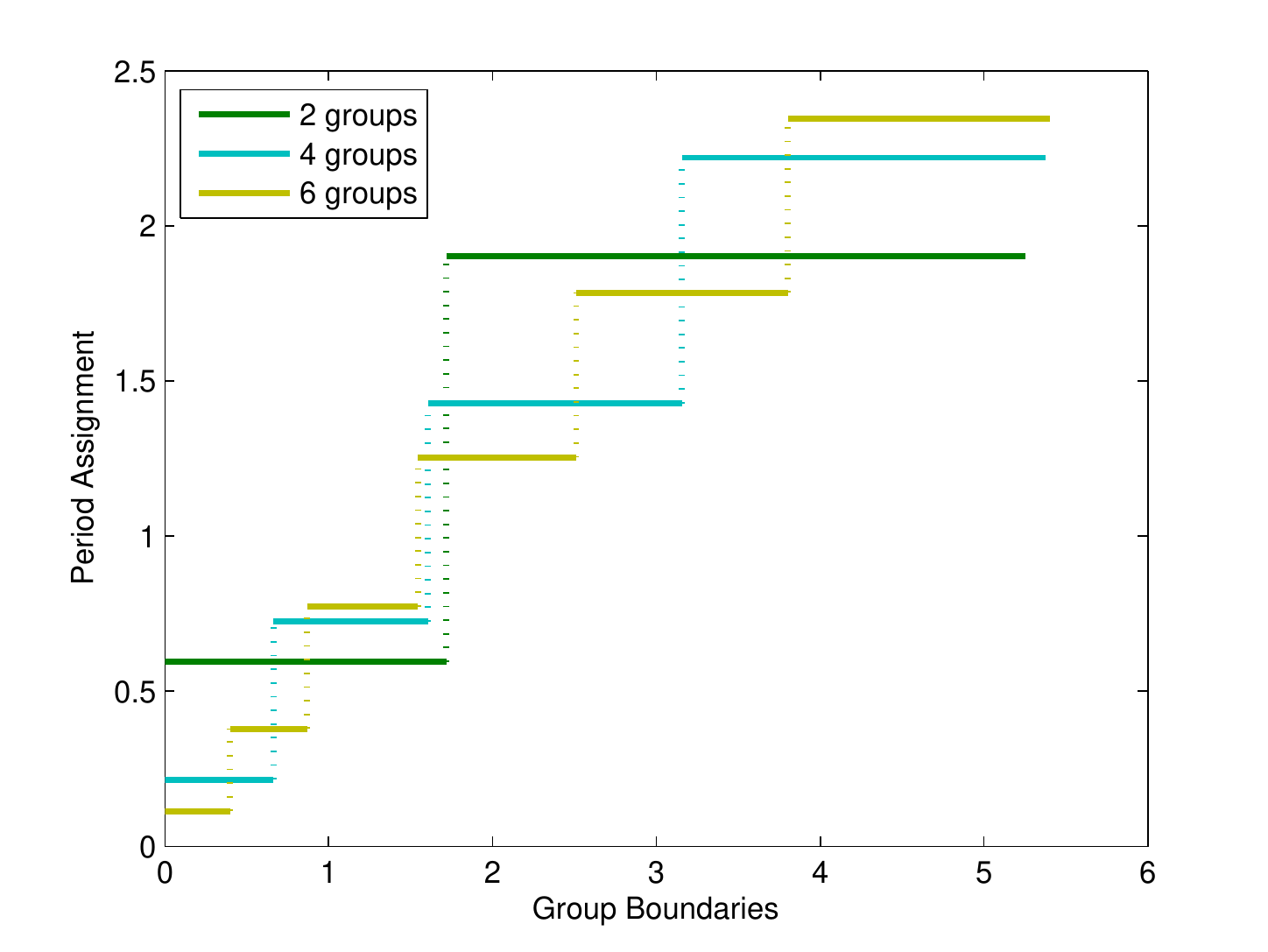}
\vspace{-3mm}
	\caption{Period assignments in continuous-consumer-type model.}
	\label{fig:c-period}
\vspace{-2mm}
\end{figure}

\begin{figure}[tbp]
	\centering
		\includegraphics[width=83mm]{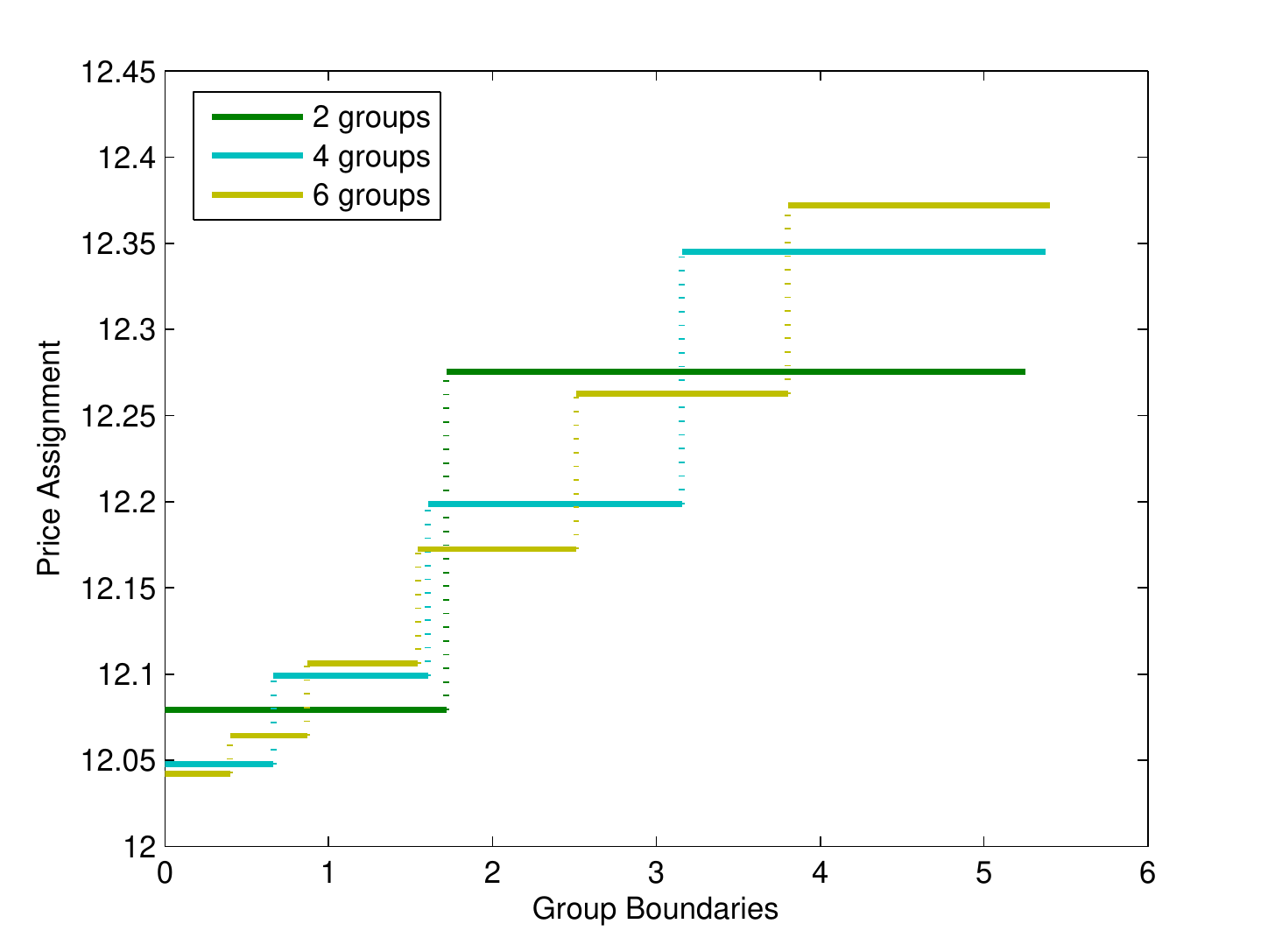}
\vspace{-3mm}
	\caption{Price assignments in continuous-consumer-type model.}
	\label{fig:c-price}
\vspace{-2mm}
\end{figure}

\subsection{Continuous-Consumer-Type}
Next, we implement the proposed contract for continuous-consumer-type model.
Figure \ref{fig:c-period} and Figure \ref{fig:c-price} show the period and price assignments in the contract with different numbers of contract items. The x-axis represents the group boundaries.
In these two figures, the distribution of consumer types (i.e., $g(\sigma)$) follows uniform distribution with $\sigma_{min}=0$ and $\sigma_{max}=6$.

From the figures we can find that as the number of group increases, more consumers are under served.
From IP Property in Proposition \ref{pop:proposition1}, with a given period length increment, the consumers with larger consumer type have a larger valuation increment than the consumers with smaller types.
Therefore, with more groups, the SP can serve consumers with a wider range, increase the price of contract item designed for consumers with large types and decrease the cost of serving the consumers with small types by increasing period assigned to consumers with large types and decreasing the period assigned to consumers with small types.
In all these ways, the SP can increase its overall profit.

\begin{figure}[tbp]
	\centering
		\includegraphics[width=92mm]{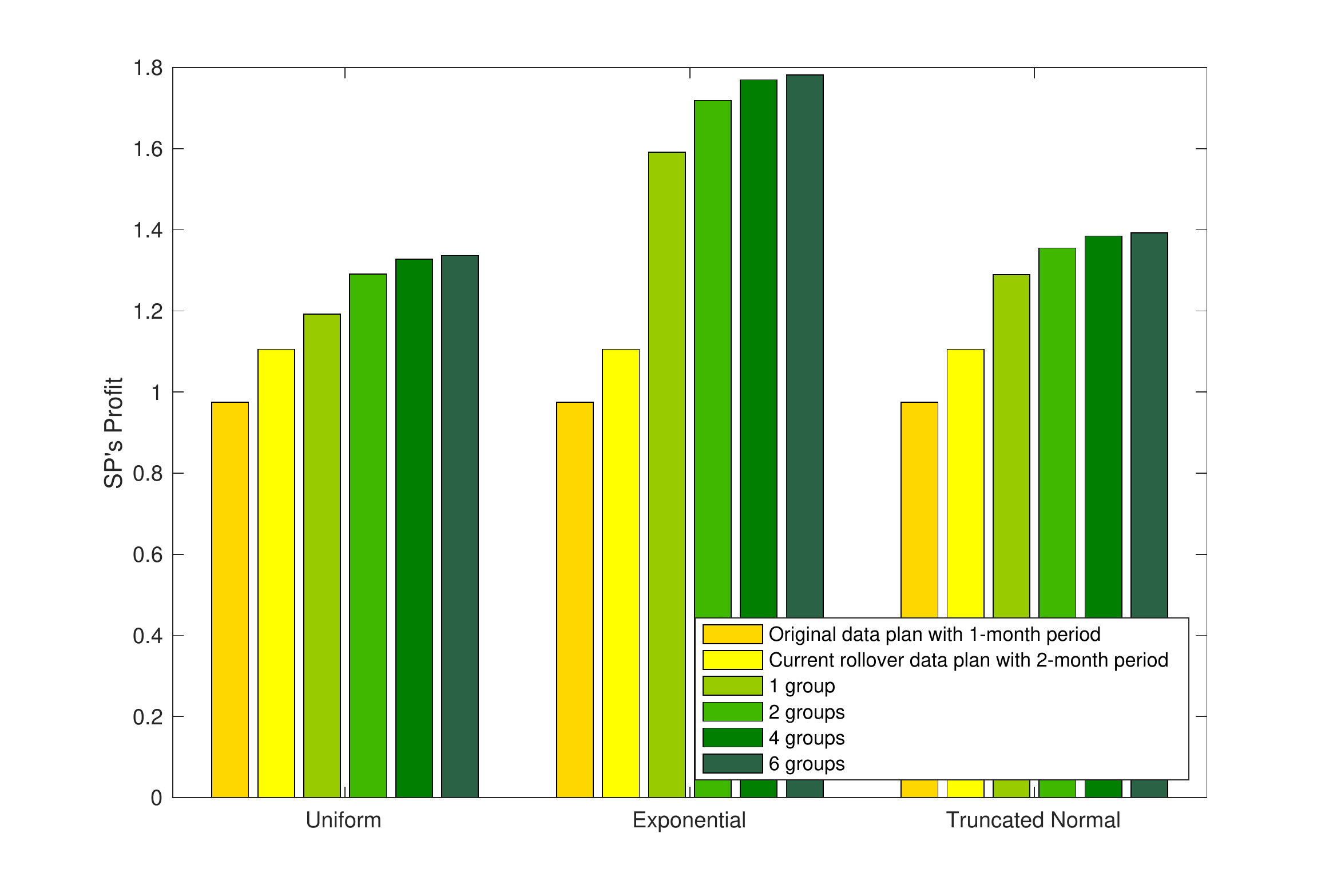}
\vspace{-6mm}
	\caption{Comparison among the original data plan with $1$-month period, current rollover data plan with $2$-month period and optimal contract in continuous-consumer-type model with different distributions of consumer types.}
	\label{fig:c-spproft}
\vspace{-3mm}
\end{figure}

Figure \ref{fig:c-spproft} represents the comparison between our proposed contracts with different group numbers and two state-of-art data plans in terms of SP's profit.
They are the original one-month data plan and the rollover data plan provided by AT\&T.
The rollover scheme can be seen as a contract with period of two months.
We run the simulation under the cases of 1) uniformly distributed consumer types, 2) exponentially distributed consumer types and 3) truncated normally distributed consumer types.
Comparing with the original data plan with $1$-month period, our contract with $6$ groups can increase the SP's profit by $37\%$, $92\%$ and $43\%$ in case 1), case 2) and case 3), respectively.
Besides, comparing with the rollover data plan with $2$-month period, our contract with $6$ groups can increase the SP's profit by $21\%$, $61\%$ and $26\%$ in case 1), case 2) and case 3), respectively.
From the figure we can find that the SP's profit is increasing with the number of groups.
Specifically, by increasing the number of groups from $1$ to $6$, the SP can increase its profit by $12$\%, $12$\% and $8$\% in case 1), case 2) and case 3), respectively;
by increasing the number of groups from $2$ to $6$, the SP can increase its profit by by $3.6$\%, $4.9$\% and $3.2$\% in case 1), case 2) and case 3), respectively.
Therefore, the SP needs to optimize the number of contract items instead of just offering one or two contract items.
Moreover, the increment from $4$ groups to $6$ groups is very small, i.e., the profit of the SP achieved by the contract with $4$ groups can reach over $98\%$ of that by the contract with $6$ groups.
Since more groups will lead to more operation costs and is not consumer friendly, we can conclude that $4$ groups is a suitable choice of group numbers, and it is in line with the number of contract items in real life \cite{CMHK}.

\section{Conclusion}
\label{sec:conclusion}
In this paper, we introduce the design of data plans with different lengths of period, in order to provide more time flexibility to consumers and increase the SP's profit.
We design a contract that contains a set of period-price combinations for each of the discrete-consumer-type model and the continuous-consumer-type model.
In the discrete-consumer-type model, each combination is intended for a consumer type, while in continuous-consumer-type model, each combination is intended for a range of consumer types.
We design the IC and IR constraints of the contracts, under which the consumer will select the contract item designed for him rather than the others.
We find the sufficient and necessary conditions of the feasible contract and design an optimal (sub-optimal) contract for the SP to maximize its overall profit in each model.
In the future, we will extend this work for two perspectives.
The first one is the user mobility. We will consider a city-wise SP, who can deploy service in several  cities, then the SP's price differentiation problem in different cities needs to be considered based on user mobilities.
The second one is the multi-dimensional data plan setting. For example, the SP can offer a two-dimensional contract, which includes both the length of the data period  and the volume of the data cap.

\section{Acknowledgments}

This work is supported by the National Natural Science
Foundation of China (Grant No. 61771162). Lin Gao
is the corresponding
author.

\bibliographystyle{ieeetr}	


\begin{thebibliography}{10}

\bibitem{7833585}
Y.~Wei, J.~Yu, T.~M. Lok, and L.~Gao, ``{A novel mobile data plan design from
  the perspective of data period},'' in {\em Proc. of IEEE ICCS}, Dec 2016.

\bibitem{joe2015mobile}
C.~Joe-Wong, S.~Ha, S.~Sen, and M.~Chiang, ``{Do mobile data plans affect
  usage? results from a pricing trial with isp customers},'' in {\em Passive
  and Active Measurement}, pp.~96--108, Springer, 2015.

\bibitem{sen2013survey}
S.~Sen, C.~Joe-Wong, S.~Ha, and M.~Chiang, ``A survey of smart data pricing:
  Past proposals, current plans, and future trends,'' {\em ACM Computing
  Surveys (CSUR)}, vol.~46, no.~2, 2013.

\bibitem{6849296}
Y.~Jin and Z.~Pang, ``Smart data pricing: To share or not to share?,'' in {\em
  Proc. of IEEE INFOCOM Workshop Smart Data Pricing (SDP)}, pp.~583--588, April
  2014.

\bibitem{6848090}
T.~Yu, Z.~Zhou, D.~Zhang, X.~Wang, Y.~Liu, and S.~Lu, ``Indapson: An incentive
  data plan sharing system based on self-organizing network,'' in {\em Proc. of
  IEEE INFOCOM}, pp.~1545--1553, April 2014.

\bibitem{6562872}
M.~Andrews, U.~Ozen, M.~I. Reiman, and Q.~Wang, ``Economic models of sponsored
  content in wireless networks with uncertain demand,'' in {\em Proc. of IEEE
  INFOCOM Workshop Smart Data Pricing (SDP)}, pp.~345--350, April 2013.

\bibitem{6849295}
L.~Zhang and D.~Wang, ``Sponsoring content: Motivation and pitfalls for content
  service providers,'' in {\em Proc. of IEEE INFOCOM Workshop Smart Data
  Pricing (SDP)}, pp.~577--582, April 2014.

\bibitem{7218537}
L.~Zheng, C.~Joe-Wong, C.~W. Tan, S.~Ha, and M.~Chiang, ``Secondary markets for
  mobile data: Feasibility and benefits of traded data plans,'' in {\em Proc.
  of IEEE INFOCOM}, pp.~1580--1588, April 2015.

\bibitem{7151094}
J.~Yu, M.~H. Cheung, J.~Huang, and H.~Poor, ``Mobile data trading: A behavioral
  economics perspective,'' in {\em Proc. of IEEE  WiOpt}, pp.~363--370, May 2015.

\bibitem{TP-toolbox-web}
AT\&T, ``{Rollover data plan}.''
\newblock {[Online]. Available:
  https://www.\\att.com/shop/wireless/rollover-data.html}.

\bibitem{Tmobile}
T-Mobile, ``{T-Mobile rollover data plan}.''
\newblock {[Online]. Available:
  http://www.t-mobile.com/offer/data-stash-data-roll.html}.

\bibitem{6342942}
L.~Duan, L.~Gao, and J.~Huang, ``{Cooperative Spectrum Sharing: A
  Contract-Based Approach},'' {\em IEEE Transactions on Mobile Computing},
  vol.~13, no.~1, pp.~174--187, 2014.

\bibitem{6464648}
L.~Gao, J.~Huang, Y.~J. Chen, and B.~Shou, ``An integrated contract and auction
  design for secondary spectrum trading,'' {\em IEEE Journal on Selected Areas
  in Communications}, vol.~31, no.~3, pp.~581--592, 2013.

\bibitem{bolton2005contract}
P.~Bolton and M.~Dewatripont, {\em Contract theory}.
\newblock MIT press, 2005.

\bibitem{li2014dynamic}
S.~Li, J.~Huang, and S.-Y.~R. Li, ``{Dynamic profit maximization of cognitive
  mobile virtual network operator},'' {\em IEEE Transactions on Mobile
  Computing}, vol.~13, no.~3, pp.~526--540, 2014.

\bibitem{7562159}
L.~Zheng and C.~Joe-Wong, ``{Understanding rollover data},'' in {\em Proc. of
  IEEE INFOCOM Workshop Smart Data Pricing (SDP)}, April 2016.

\bibitem{shidi}
Z.~Wang, L.~Gao, and J.~Huang, ``{Pricing Optimization of Rollover Data
  Plan},'' in {\em Proc. of IEEE WiOpt}, May 2017.

\bibitem{shidi2}
Z.~Wang, L.~Gao, and J.~Huang, ``{A Contract-Theoretic Design of Mobile Data
  Plan with Time Flexibility},'' in {\em Proc. of ACM NetEcon}, June 2017.

\bibitem{add-1}
Z.~Wang, Lin~Gao, and J.~Huang,
``Multi-Dimensional Contract Design for Mobile Data Plan with Time Flexibility,''
in {\em Proc. of ACM MobiHoc}, June 2018.

\bibitem{add-2}
Z.~Wang, Lin~Gao, and J.~Huang,
``Duopoly Competition for Mobile Data Plans with Time Flexibility,''
in {\em Proc. of IEEE WiOpt}, May 2018.



\bibitem{b13}
Z.~Bao, W.~Qiu, L.~Wu, F.~Zhai, W.~Xu, B.~Li, and Z.~Li, ``{Optimal
  Multi-Timescale Demand Side Scheduling Considering Dynamic Scenarios of
  Electricity Demand},'' {\em IEEE Transactions on Smart Grid}, vol.~PP,
  no.~99, pp.~1--1, 2018.

\bibitem{b14}
J.~Ding, R.~Xu, Y.~Li, P.~Hui, and D.~Jin, ``{Measurement-driven Modeling for
  Connection Density and Traffic Distribution in Large-scale Urban Mobile
  Networks},'' {\em IEEE Transactions on Mobile Computing}, vol.~PP, no.~99,
  pp.~1--1, 2018.

\bibitem{a2}
Y.~Zhang, L.~Song, M.~Pan, Z.~Dawy, and Z.~Han, ``{Non-Cash Auction for
  Spectrum Trading in Cognitive Radio Networks: Contract Theoretical Model With
  Joint Adverse Selection and Moral Hazard},'' {\em IEEE Journal on Selected
  Areas in Communications}, vol.~35, no.~3, pp.~643--653, 2017.

\bibitem{b1}
R.~T. Ma, ``{Usage-based pricing and competition in congestible network service
  markets},'' {\em IEEE/ACM Transactions on Networking}, vol.~24, no.~5,
  pp.~3084--3097, 2016.

\bibitem{5738226}
L.~Gao, X.~Wang, Y.~Xu, and Q.~Zhang, ``Spectrum trading in cognitive radio
  networks: A contract-theoretic modeling approach,'' {\em IEEE Journal on
  Selected Areas in Communications}, vol.~29, no.~4, pp.~843--855, 2011.

\bibitem{a8}
W.~Dai and S.~Jordan, ``{Design and impact of data caps},'' in {\em Proc. of
  IEEE GLOBECOM}, pp.~1650--1656, December 2013.

\bibitem{a9}
W.~Dai and S.~Jordan, ``{The effect of data caps upon isp service tier design
  and users},'' {\em ACM Transactions on Internet Technology (TOIT)}, vol.~15,
  no.~2, p.~8, 2015.

\bibitem{a10}
X.~Wang, R.~T. Ma, and Y.~Xu, ``{The role of data cap in two-part pricing under
  market competition},'' in {\em Proc. of IEEE INFOCOM Workshop Computer
  Communications}, April 2015.

\bibitem{a11}
L.~Duan, J.~Huang, and B.~Shou, ``{Duopoly competition in dynamic spectrum
  leasing and pricing},'' {\em IEEE Transactions on Mobile Computing}, vol.~11,
  no.~11, pp.~1706--1719, 2012.

\bibitem{a12}
Y.~Luo, L.~Gao, and J.~Huang, ``{An integrated spectrum and information market
  for green cognitive communications},'' {\em IEEE Journal on Selected Areas in
  Communications}, vol.~34, no.~12, pp.~3326--3338, 2016.

\bibitem{CMHK}
CMHK, ``{Local Service Plan}.''
\newblock {[Online]. Available:
  http://www.hk.\\chinamobile.com/en/corporate\_information/Service\_Plans/4.5G\\\_Service\_Plan/4Glocal\_serviceplan.html}.

\bibitem{report}
Y.~Wei, J.~Yu, T.-M.~Lok, and L.~Gao, ``A Novel Mobile Data Contract Design with Time
Flexibility,'' \emph{Technical Report},
[Online]. Available: http://arxiv.org/abs/1806.07308



\end{thebibliography}

\begin{IEEEbiography}
[{\includegraphics[width=1in,height=1.25in,clip,keepaspectratio]{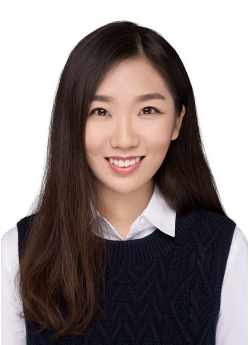}}]
{Yi Wei} (S'14) received her Ph.D. degree in the Department of Information Engineering at the Chinese University of Hong Kong in 2017. Her research interests include smart data pricing and interference alignment in wireless communication. She is a student member of IEEE.
\end{IEEEbiography}

\begin{IEEEbiography}
[{\includegraphics[width=1in,height=1.25in,clip,keepaspectratio]{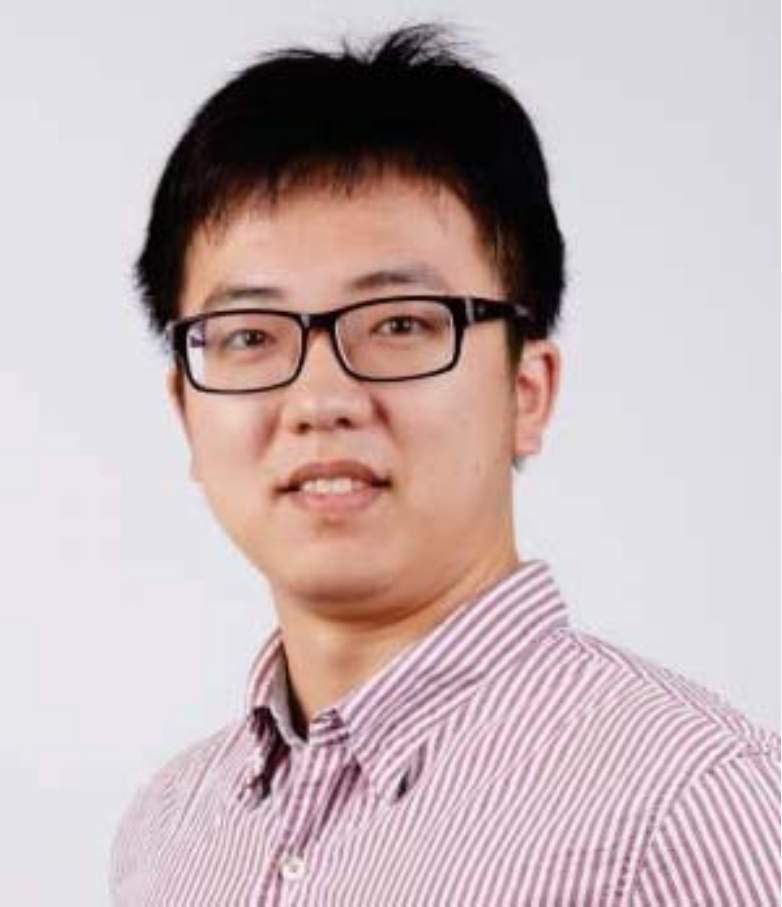}}]
{Junlin Yu} (S'14) received his Ph.D. degree in the Department of Information Engineering at the Chinese University of Hong Kong in 2017. His research interests include behavioral economical studies in wireless communication networks and optimization in mobile data trading. He is a student member of IEEE.
\end{IEEEbiography}

\begin{IEEEbiography}
[{\includegraphics[width=1in,height=1.25in,clip,keepaspectratio]{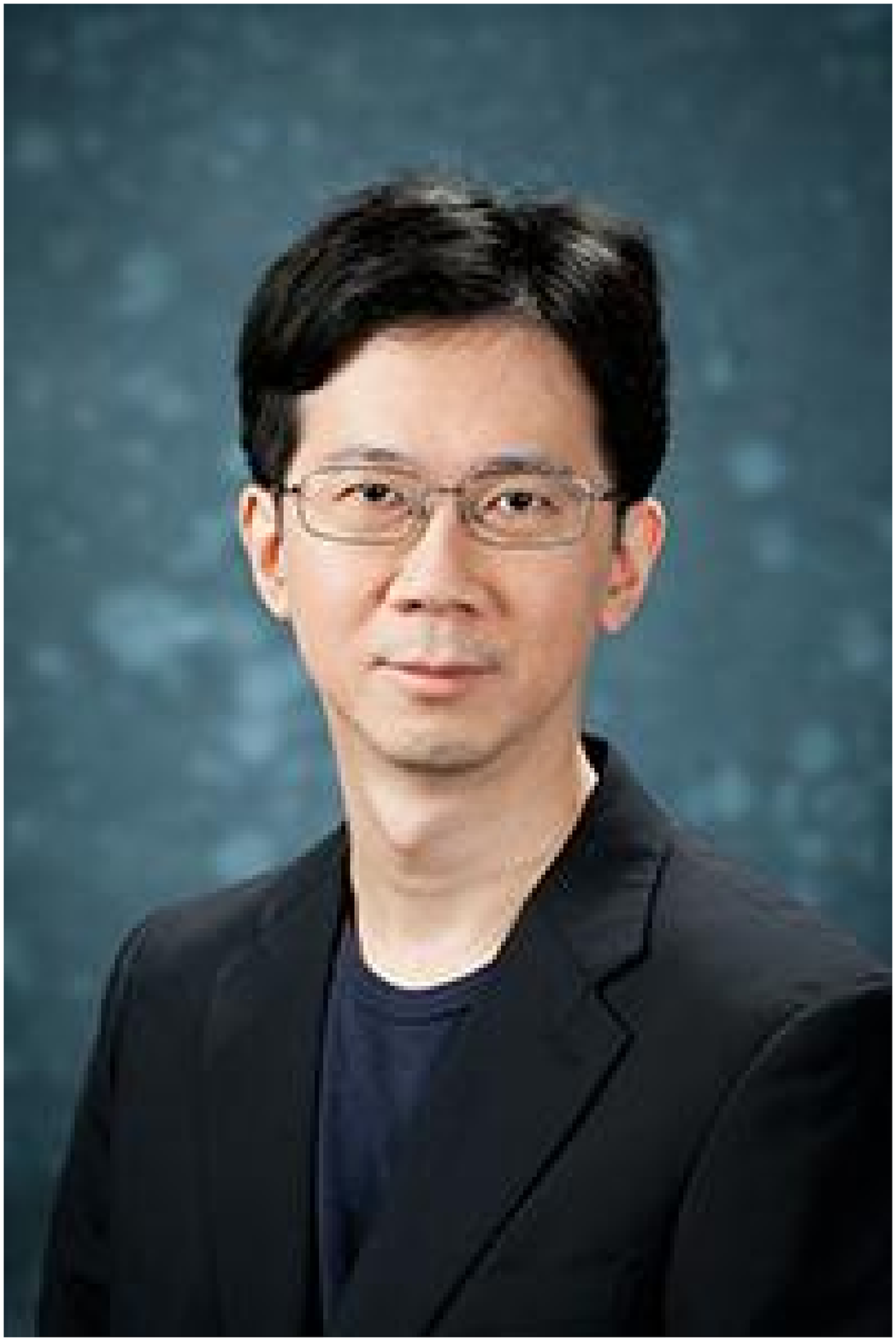}}]
{Tat-Ming Lok} (SM'03) received the B.Sc. degree in electronic engineering from the Chinese University of Hong Kong, Shatin, Hong Kong, in 1991 and the M.S.E.E. and Ph.D. degrees in electrical engineering from Purdue University, West Lafayette, IN, USA in 1992 and 1995 respectively. He was a Postdoctoral Research Associate with Purdue University. He then joined the Chinese University of Hong Kong, where he is currently an Associate Professor. His research interests include communication theory, communication networks, signal processing for communications, and wireless systems. He has served on Technical Program Committees of different international conferences, including the IEEE International Conference on Communications, IEEE Vehicular Technology Conference, IEEE Globecom, IEEE Wireless Communications and Networking Conference, and IEEE International Symposium on Information Theory. He was a co-chair of the Wireless Access Track of the IEEE Vehicular Technology Conference in 2004. He also served as an Associate Editor for the IEEE TRANSACTIONS ON VEHICULAR TECHNOLOGY from 2002 to 2008. He has been serving as an Editor for the IEEE TRANSACTIONS ON WIRELESS COMMUNICATIONS since 2015.
\end{IEEEbiography}

\begin{IEEEbiography}
[{\includegraphics[width=1in,height=1.25in,clip,keepaspectratio]{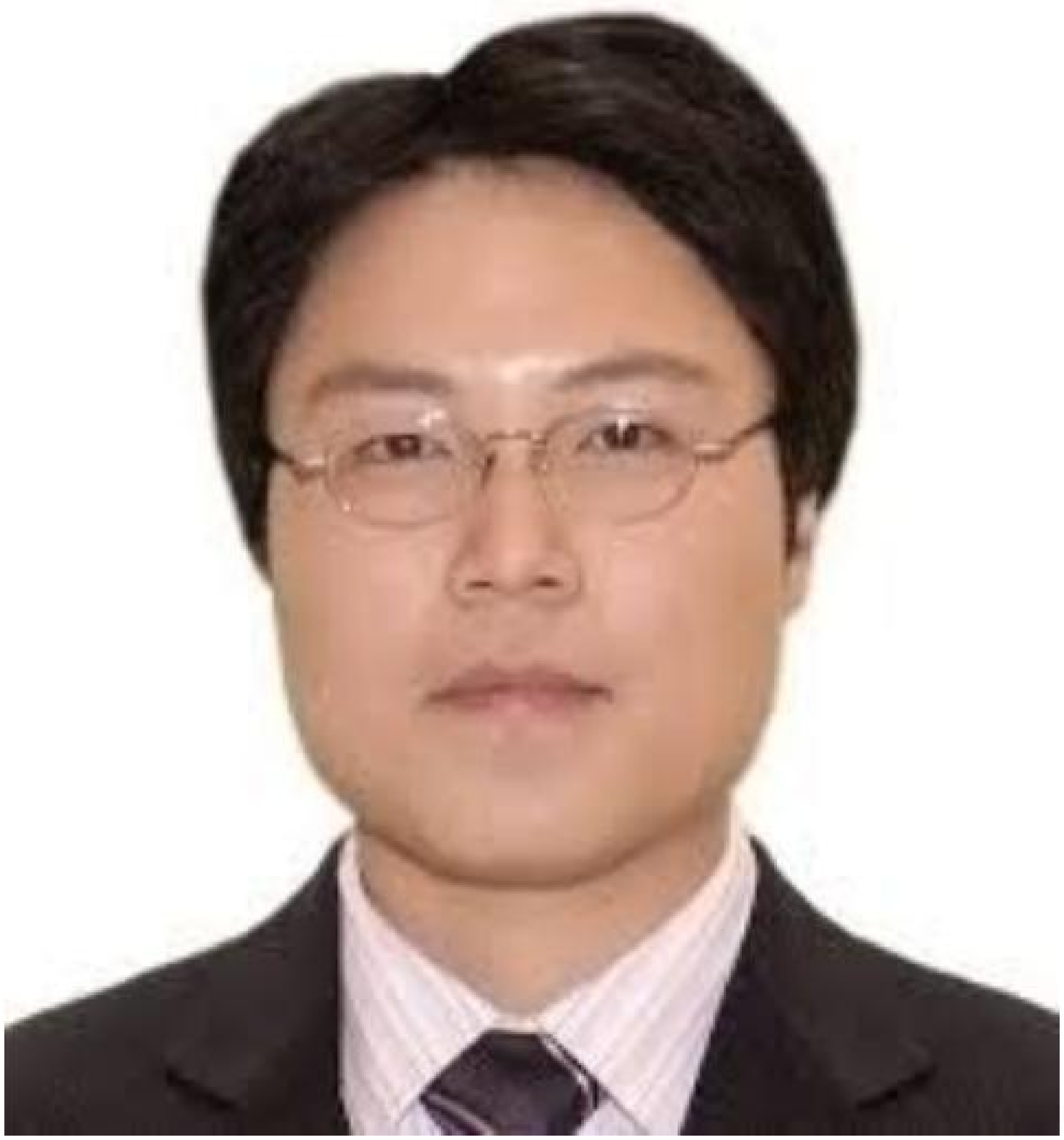}}]
{Lin Gao} (S'08-M'10-SM'16)
is an Associate Professor
with the School of Electronic and Information
Engineering, Harbin Institute of Technology,
Shenzhen, China. He received the Ph.D.
degree in Electronic Engineering from Shanghai
Jiao Tong University in 2010.
His main research
interests are in the area of network economics
and games, with applications in wireless communications
and networking.
He was a co-recipient of three Best Paper Awards from WiOpt 2013, 2014, 2015, and one Best Paper Award Finalist from IEEE INFOCOM 2016.
He received the
IEEE ComSoc Asia-Pacific Outstanding Young
Researcher Award in 2016.
\end{IEEEbiography}

\clearpage

\thispagestyle{empty}

\appendices

\section{Proof of Proposition \ref{pop:proposition1}}
\label{sec:proofpp1}

	First, we have
	\begin{align}
	    V_{\sigma t}(\sigma,t) &=\frac{\alpha}{2t^{1.5}}\!\int_{\frac{\sqrt{t}\Delta{q}}{\sigma}}^{+\infty}xf(x|0,1)\ud{x} + \frac{{\alpha\Delta{q}}^2}{2\sqrt{t}\sigma^2}f(\frac{\sqrt{t}\Delta{q}}{\sigma}|0,1)\notag\\
	    &>0.\notag
	\end{align}
	 Since $\sigma > \sigma'$ and $t > t'$, we can rewrite \eqref{eq:ip1} as:
	\begin{align}
	&V(\sigma,t)-V(\sigma,t')-\Big(V(\sigma',t)-V(\sigma',t')\Big) \notag\\
	=&\int_{t'}^{t}V_{t}(\sigma,x)\ud{x} - \int_{t'}^{t}V_{t}(\sigma',x)\ud{x}\notag\\
	=&\int_{t'}^{t}\Big(\int_{\sigma'}^{\sigma}V_{\sigma t}(y,x)\ud{y}\Big)\ud{x}>0.\notag
	\end{align}

\section{Proof of Lemma \ref{lm:lemma1}}
\label{sec:prooflm1}

We prove the lemma by contradiction. If $\sigma_i > \sigma_j$ and $t_j>t_i$ hold at the same time, then from the IP property we have:
\begin{align}
&V(\sigma_i,t_j) -V(\sigma_i,t_i)>V(\sigma_j,t_j)-V(\sigma_j,t_i)\notag\\
\Rightarrow &V(\sigma_i,t_j)+V(\sigma_j,t_i)>V(\sigma_i,t_i)+V(\sigma_j,t_j),\notag
\end{align}
which violates the IC constraint:
\begin{displaymath}
\begin{array}{ll}
&\begin{array}{ll}
V(\sigma_i,t_i) - \pi_i \geq V(\sigma_i,t_j)-\pi_j\\
V(\sigma_j,t_j) - \pi_j \geq V(\sigma_j,t_i)-\pi_i
\end{array} \Big\}\Rightarrow \\
&V(\sigma_i,t_i) +V(\sigma_j,t_j) \geq V(\sigma_i,t_j)+V(\sigma_j,t_i).\notag
\end{array}
\end{displaymath}

\section{Proof of Lemma \ref{lm:lemma2}}
\label{sec:prooflm2}

\begin{proof} We prove the right direction first and then the left:
\begin{enumerate}
	\item
	From the IC constraint, if $t_i > t_j$, we have:
	\begin{align}
	&V(\sigma_j,t_j)-\pi_j\geq V(\sigma_j,t_i)-\pi_i\notag\\
	\Rightarrow &V(\sigma_j,t_j)-V(\sigma_j,t_i) \geq \pi_j - \pi_i.\notag
	\end{align}
	Since $V_t(\sigma,t) >0$ and $t_i>t_j$, we can find:
	\begin{equation}
    V(\sigma_j,t_j) - V(\sigma_j,t_i)<0
    \Rightarrow \pi_j -\pi_i < 0.\notag
	\end{equation}
	
	\item
	From the IC constraint,  we have:
	\begin{align}
	&V(\sigma_i,t_i)-\pi_i\geq V(\sigma_i,t_j)-\pi_j\notag\\
	\Rightarrow &V(\sigma_i,t_i)-V(\sigma_i,t_j) \geq \pi_i - \pi_j.\notag
	\end{align}
	If $\pi_i > \pi_j$, then
	\begin{equation}
	V(\sigma_i,t_i)>V(\sigma_i,t_j).\notag
	\end{equation}
	Since $V_t(\sigma,t) >0$ and $V(\sigma_i,t_i)>V(\sigma_i,t_j)$, we can find $t_i > t_j$.
\end{enumerate}
\end{proof}

\section{Proof of Theorem \ref{th:theorem1}}
\label{sec:app1}
We first prove the sufficiency of the conditions in Theorem \ref{th:theorem1} and then the necessity of them.
\begin{enumerate}
	\item Sufficiency.

	We can use mathematical induction to prove the sufficiency of the conditions. We use $\mathscr{C}_d(i)$ to denote the subset of contract $\mathscr{C}_d$ containing the last $i$ contract items, i.e., $\mathscr{C}_d(i) = \{(\sigma_j,t_j)|j =I-i+1,\ldots,I \}$.

	We first prove $\mathscr{C}_d(i)$ is feasible. Since there is only one consumer type $\sigma_I$ in this contract, we only need to justify the IR constraint, which can be proved directly from the condition \eqref{eq:t12} in Theorem \ref{th:theorem1}.

	Then, we prove if $\mathscr{C}_d(i)$ is feasible, $\mathscr{C}_d(i+1)$ is also feasible. We have the following conditions if $\mathscr{C}_d(i)$ is feasible.
	\begin{align}
	&V\!(\sigma_{I\!-\!i\!+\!1},\!t_{I\!-\!i\!+\!1}) \!-\!\pi_{I\!-\!i\!+\!1}\! \geq \! V(\sigma_{I\!-\!i\!+\!1},\!t_m)\!-\! \pi_m,\!\forall m \in \!\mathcal{I}_{i}\footnotemark, \label{eq:c1}\\
	&V\!(\sigma_{m},\!t_{m}) \!-\!\pi_{m}\! \geq\!  V\!(\sigma_{m},\!t_{I\!-\!i\!+\!1})\!-\! \pi_{I\!-\!i\!+\!1},\!\forall m\in \!\mathcal{I}_{i}, \label{eq:c5}\\
	&V(\sigma_m,t_m) - \pi_m \geq 0,~~\forall m\in \mathcal{I}_{i}.\label{eq:c2}
	\end{align}
    From the conditions \eqref{eq:t13} and \eqref{eq:t14} of Theorem \ref{th:theorem1}, we have:
	\begin{align}
	&\pi_{I\!-\!i}\! \geq\! \pi_{I\!-\!i+1}\! +\! V\!(\sigma_{I\!-\!i+1},\!t_{I\!-\!i})\!-\! V\!(\sigma_{I\!-\!i+1},\!t_{I\!-\!i+1}),\label{eq:c4}\\
	&\pi_{I\!-\!i}\! \leq\! \pi_{I\!-\!i+1}\! +\! V(\sigma_{I\!-\!i},t_{I\!-\!i})\!-\! V(\sigma_{I-i},t_{I\!-\!i+1}).\label{eq:c3}
	\end{align}
	\footnotetext{The set $\mathcal{I}_{i}=\{I-i+1,I-i+2\ldots,I\}$.}
	With the above conditions and the IP property, we are going to prove that the IC and IR constraints for the contract $\mathscr{C}_d(i+1)$ are satisfied.

	IC constraints:
	\begin{align}
	&V\!(\sigma_{I-i},\!t_{I-\!i})\!-\!\pi_{I-\!i} \!\geq\! V\!(\sigma_{I\!-\!i},\!t_m)\! -\! \pi_m, \forall m \in\! \mathcal{I}_{i}, \label{eq:r2}\\
	&V\!(\sigma_m,\! t_m\!)\! -\! \pi_m \!\geq\! V\!(\sigma_m,\! t_{I-i})\! -\! \pi_{I-i}, \forall m \in\! \mathcal{I}_{i}, \label{eq:r3}
	\end{align}
	IR constraint:
	\begin{equation}
	V\!(\sigma_{I-i},\!t_{I-\!i})\!-\! \pi_{I-\!i} \geq 0. \label{eq:r1}
	\end{equation}
	If the above constrains are satisfied, $\mathscr{C}_d(i+1)$ is feasible.

    By adding up \eqref{eq:c1} and \eqref{eq:c3}, we have:
    \begin{align}
    V\!(\sigma_{I\!-\!i},t_{I\!-\!i}) \!-\! \pi_{I\!-\!i}
    \!\geq\! V\!(&\sigma_{I\!-\!i},t_{I\!-\!i+\!1}) \!+\!V\!(\sigma_{I-i+1},t_m)\notag\\
    &-V(\sigma_{I\!-\!i+\!1}, t_{I\!-\!i+\!1})\!-\! \pi_{m},\notag
    \end{align}
    for all $m \!\in \!\mathcal{I}_{i}$. From the IP property, we have:
    \begin{align}
    V(\sigma_{I-i+1},t_m)&-V(\sigma_{I-i+1}, t_{I-i+1})\notag\\
    \geq &V(\sigma_{I-i},t_m)-V(\sigma_{I-i},t_{I-i+1}),\notag
    \end{align}
    for all $m \in \mathcal{I}_{i}$, since $\sigma_{I-i+1} > \sigma_{I-i}$ and $t_{m} \geq t_{I-i+1}$. By adding up the above two equations, \eqref{eq:r2} is proved.

    By adding up \eqref{eq:c5} and \eqref{eq:c4}, we have:
    \begin{align}
    V\!(\sigma_m,t_m) - \pi_m
    \geq\! V\!(&\sigma_m,t_{I-i+1}) \!+\!V\!(\sigma_{I-i+1},t_{I-i})\notag\\
    &-V(\sigma_{I-i+1}, t_{I-i+1})- \pi_{I-i},\notag
    \end{align}
    for all $m\! \in\! \mathcal{I}_{i}$. From the IP property, we have:
    \begin{align}
    V(\sigma_m,t_{I-i+1})&-V(\sigma_m, t_{I-i})\notag\\
    \geq &V(\sigma_{I-i+1},t_{I-i+1})-V(\sigma_{I-i+1},t_{I-i}),\notag
    \end{align}
    for all $m \in \mathcal{I}_{i}$, since $\sigma_m \geq \sigma_{I-i+1}$ and $t_{I-i+1} \geq t_{I-i}$. By adding up the above two equations, \eqref{eq:r3} is proved.

    From \eqref{eq:r2}, \eqref{eq:c2} and the property $V_{\sigma}(\sigma,t) < 0$, we have:
    \begin{align}
    V(\sigma_{I-i},t_{I-i}) -\pi_{I-i}
    \geq &V(\sigma_{I-i},t_m) - \pi_m\notag\\
    \geq &V(\sigma_m,t_m) - \pi_m \geq 0,\notag
    \end{align}
    and \eqref{eq:r1} is proved.

	\item Necessity.

	Lemma \ref{lm:lemma1} shows the necessity of the condition \eqref{eq:t11} in Theorem \ref{th:theorem1}. The condition \eqref{eq:t12} in Theorem \ref{th:theorem1} can be derived from the IR constraint. The conditions \eqref{eq:t13} and \eqref{eq:t14} can be proved by the IC constraints for types $\sigma_i$ and $\sigma_{i+1}$ (i.e., $V(\sigma_{i+1}, t_{i+1})-\pi_{i+1} \geq V(\sigma_{i+1}-\pi_i)$ and $V(\sigma_i, t_i)-\pi_i \geq V(\sigma_i, t_{i+1})-\pi_{i+1}$).

\end{enumerate}

\section{Proof of Theorem \ref{th:theorem2}}
\label{sec:app2}
Condition \eqref{eq:t234} is equivalent to the following two conditions
\begin{align}
&\pi_{k} \geq \pi_{k+1} + V(\sigma_{k}^{[max]},t_k)- V(\sigma_{k}^{[max]},t_{k+1}).\label{eq:t23}\\
&\pi_{k} \leq \pi_{k+1} + V(\sigma_{k}^{[max]},t_k)- V(\sigma_{k}^{[max]},t_{k+1}).\label{eq:t24}
\end{align}
We first prove that the constraints in \eqref{eq:t23} and \eqref{eq:t24} are sufficient and necessary conditions for the following two constraints:
\begin{align}
&\pi_{k} \geq \pi_{k+1} \!+\! V(\sigma,t_k)\!-\! V(\sigma,t_{k+1}), ~\forall \sigma \in [\sigma_{k}^{[max]}, \sigma_{k+1}^{[max]}].\label{eq:t3}\\
&\pi_{k} \leq \pi_{k+1} \!+\! V(\sigma,t_k)\!-\! V(\sigma,t_{k+1}),  ~\forall \sigma \in [\sigma_{k-1}^{[max]}, \sigma_{k}^{[max]}].\label{eq:t4}
\end{align}
\begin{enumerate}
	\item Sufficiency.
	
From \eqref{eq:t23}, we have
\begin{align}
V(\sigma_{k}^{[max]}, t_{k+1})-V(\sigma_{k}^{[max]},t_k)\geq \pi_{k+1}-\pi_{k}.\notag
\end{align}
The IP property implies that
\begin{align}
&V(\sigma,t_{k+1})-V(\sigma,t_k)\notag\\
&~~~~~~~\geq V(\sigma_{k}^{[max]},t_{k+1})-V(\sigma_{k}^{[max]},t_k),\notag\\
&~~~~~~~~~~~~~~~~~~~~~~~~~~~~~~~~~~~~~\forall \sigma \in [\sigma_{k}^{[max]}, \sigma_{k+1}^{[max]}],\notag \\
\Rightarrow &V(\sigma, t_{k+1})-V(\sigma,t_k) \geq \pi_{k+1} - \pi_k,\notag\\
&~~~~~~~~~~~~~~~~~~~~~~~~~~~~~~~~~~~~~\forall \sigma \in [\sigma_{k}^{[max]}, \sigma_{k+1}^{[max]}].\notag
\end{align}
Hence, \eqref{eq:t3} is proved. In addition, from \eqref{eq:t24}, we have
\begin{align}
V(\sigma_{k}^{[max]}, t_{k+1})-V(\sigma_{k}^{[max]},t_k)\leq \pi_{k+1}-\pi_{k}.\notag
\end{align}
The IP property implies that
\begin{align}
&V(\sigma,t_{k+1})-V(\sigma,t_k)\notag\\
&~~~~~~\leq V(\sigma_{k}^{[max]},t_{k+1})-V(\sigma_{k}^{[max]},t_k),\notag\\
&~~~~~~~~~~~~~~~~~~~~~~~~~~~~~~~~~~~~~\forall \sigma \in [\sigma_{k-1}^{[max]}, \sigma_{k}^{[max]}],\notag \\
\Rightarrow &V(\sigma, t_{k+1})-V(\sigma,t_k)\leq\pi_{k+1}-\pi_k,\notag\\
&~~~~~~~~~~~~~~~~~~~~~~~~~~~~~~~~~~~~~\forall \sigma \in [\sigma_{k-1}^{[max]}, \sigma_{k}^{[max]}].\notag
\end{align}
Hence, \eqref{eq:t4} is proved.

From the above derivations, we know that the conditions in \eqref{eq:t23} and \eqref{eq:t24} are the sufficient conditions of \eqref{eq:t3} and \eqref{eq:t4}.
	
\item Necessity
	
	Since $\sigma_{k}^{[max]} \!\in\! [\sigma_{k}^{[max]}, \sigma_{k+1}^{[max]}]$ and $\sigma_{k}^{[max]} \in [\sigma_{k-1}^{[max]},$ $\sigma_{k}^{[max]}]$, the necessity of the conditions in \eqref{eq:t23} and \eqref{eq:t24} to the conditions in \eqref{eq:t3} and \eqref{eq:t4} can be obtained directly.

\end{enumerate}

By now, we have proved that the conditions in Theorem \ref{th:theorem2} are equivalent to the conditions \eqref{eq:t21}, \eqref{eq:t22}, \eqref{eq:t3} and \eqref{eq:t4}.
From Theorem \ref{th:theorem1}, we know that the IC and IR constraints are equivalent to the conditions in \eqref{eq:t21}, \eqref{eq:t22}, \eqref{eq:t3} and \eqref{eq:t4}. Therefore, Theorem \ref{th:theorem2} is proved.


\thispagestyle{empty}

\section{Proof of Lemma \ref{lm:lemma5}}
\label{sec:app3}
Without loss of generality, we can let $t\Delta q^2 = x \sigma^2$, where $x \geq 0$.

When $x >1$, we have $\sigma^2 < t\Delta q^2$ and the value of $\frac{t\Delta q^2(\sigma^2-t\Delta q^2)}{\sigma^3(\sigma^2+t\Delta q^2)}$ is negative. Hence, \eqref{eq:max} is satisfied.

When $x \!\in \! [0,1]$, we rewrite the formula $\frac{t\Delta q^2(\sigma^2-t\Delta q^2)}{\sigma^3(\sigma^2+t\Delta q^2)}$ as $\frac{x(1-x)}{\sigma(1+x)}$. The second order derivative of $\frac{x(1-x)}{1+x}$ is $-\frac{4}{{(1+x)}^3}$, which is negative. Hence, the optimal solution $\hat{x}$ that leads to the maximum value of $\frac{x(1-x)}{1+x}$ satisfies
\begin{equation}
\frac{\partial \frac{x(1-x)}{1+x}}{\partial x}\bigg|_{x = \hat{x}} =0 \Rightarrow \frac{1-2\hat{x}-{\hat{x}}^2}{{(1+\hat{x})}^2} = 0\Rightarrow \hat{x} = \sqrt{2}-1.\notag
\end{equation}
Therefore, the maximum value of $\frac{x(1-x)}{1+x}$ is $3-2\sqrt{2}$. In other words, the maximum value of $\frac{t\Delta q^2(\sigma^2-t\Delta q^2)}{\sigma^3(\sigma^2+t\Delta q^2)}$ is $\frac{3-2\sqrt{2}}{\sigma}$ and \eqref{eq:max} is satisfied.

\section{Proof of Proposition \ref{pop:proposition3}}
\label{sec:app4}
The first order derivative of $\sum_{k=i}^jQ_k(\sigma)$ \rtwo{with} respect to $\sigma$ is
\begin{align}
&\frac{\partial \sum_{k=i}^jQ_k(\sigma)}{\partial \sigma} = \sum_{k=i}^j\frac{\partial Q_k(\sigma)}{\partial \sigma}\notag\\
 &=Ng(\sigma)\Big(\sum_{k=i}^jH_k(\sigma)+C(t_{j+1})-C(t_{i})\Big),\notag
\end{align}
where $H_k(\sigma)$ is defined in \eqref{eq:H}.
Since the first order derivative of $H_k(\sigma)$ is non-positive for all $k$, we have $\frac{\partial \sum_{k=i}^jH_k(\sigma)}{\partial \sigma} \leq 0$.
Hence, $\sum_{k=i}^jH_k(\sigma)$ crosses zero at most once. Together with the facts that \rtwo{i) $C(t_{j+1})-C(t_i)$ is a constant with respect to $\sigma$ and ii) $Ng(\sigma)$ is always positive, we have the result that $\sum_{k=i}^jQ_k(\sigma)$ is unimodal with respect to $\sigma$.}

\section{Proof of Proposition \ref{pop:proposition4}}
\label{sec:app5}
The statement is trivial if $\hat{\sigma}_1 =\hat{\sigma}_2$, thus we focus on the case of $\hat{\sigma}_1 > \hat{\sigma}_2$.

The statement can be proved if for arbitrary $\sigma_1 < \sigma_2$, we can find a $\sigma^*$ such that $\sum_{k=1}^2Q_k(\sigma^*)> \sum_{k=1}^2Q_k(\sigma_k)$.
There are two possible cases of $\sigma_2$: 1) $\sigma_2 \geq \hat{\sigma}_1$, and 2) $\sigma_2 < \hat{\sigma}_1$.

For the case that $\sigma_2 \geq \hat{\sigma}_1$.
Since $\hat{\sigma}_1$ is the optimal solution of $Q_1(\sigma_1)$, we have $Q_1(\hat{\sigma}_1)\geq Q_1(\sigma_1)$ for any $\sigma_1$.
Since $\sigma_2>\hat{\sigma}_1 > \hat{\sigma}_2$ and $Q_2$ is a unimodal function, we have $\frac{\partial Q_2(\sigma_2)}{\partial \sigma_2} \leq 0$ for any $\sigma_2 > \hat{\sigma}_2$, which means that $Q_2(\hat{\sigma}_1)\geq Q_2(\sigma_2)$.
Therefore, by letting $\sigma^* = \hat{\sigma}_1$, we have $\sum_{k=1}^2Q_k(\sigma^*)> \sum_{k=1}^2Q_k(\sigma_k)$.

For the case that $\sigma_2 < \hat{\sigma}_1$, by letting $\sigma^* = \sigma_2$, we have $\sum_{k=1}^2Q_k(\sigma^*)> \sum_{k=1}^2Q_k(\sigma_k)$. This is because $\sigma_1<\sigma_2<\hat{\sigma}_1$ and $Q_1$ is a unimodal function, which implies that 1) $\frac{\partial Q_1(\sigma_1)}{\partial \sigma_1} \geq 0$ for any $\sigma_1 < \sigma_2$, and 2) $Q_1(\sigma^*) = Q_1(\sigma_2)\geq Q_1(\sigma_1)$.

\thispagestyle{empty}

\section{Proof of Proposition \ref{pop:proposition5}}
\label{sec:proofpp5}
First, we rewrite the formula $\frac{x}{1-e^{-x}}$ as
\begin{equation}
\frac{x}{1-e^{-x}} = 1+\frac{e^{-x}+x-1}{1-e^{-x}}.\notag
\end{equation}
Proposition \ref{pop:proposition5} is proved if $e^{-x}+x-1\geq 0$.
Since the first order derivative of $e^{-x}+x-1$ is $1-e^{-x}$, which is positive for any $x>0$, the minimum value of $e^{-x}+x-1$ is then lower bounded by $e^{-0}+0-1 = 0$.

%
%
%
%

\end{document}